\documentclass[reqno,11pt]{amsart}
\usepackage{amsmath,amssymb,dsfont}
\usepackage{graphicx}
\usepackage[a4paper,  margin=3.3cm]{geometry}

\newtheorem{theorem}{Theorem}[section]
\newtheorem{corollary}[theorem]{Corollary}
\newtheorem{lemma}[theorem]{Lemma}
\newtheorem{proposition}[theorem]{Proposition}
\theoremstyle{definition}

\theoremstyle{remark} \theoremstyle{remark}

\newtheorem{remark}[theorem]{Remark}
\numberwithin{equation}{section}

\allowdisplaybreaks[3]

\title[A Phase Transition in a Widom-Rowlinson Model]{A Phase Transition in a Widom-Rowlinson Model with Curie-Weiss Interaction}

\author{ Yuri  Kozitsky}

\address{Instytut Matematyki, Uniwersytet Marii Curie-Sk{\l}odowskiej, 20-031 Lublin, Poland}
\email{jkozi@hektor.umcs.lublin.pl}

\author{Mykhailo Kozlovskii}

\address{Institute for Condensed Matter Physics/ Lviv, Ukraine}
\email{mpk@icmp.lviv.ua}

\begin{document}

\begin{abstract}

An analog of the continuum Widom-Rowlinson model is introduced and
studied. Its two-component version is a gas of point particles of
types 0 and 1 placed in $\mathds{R}^d$, in which like particles do
not interact and unlike particles contained in a vessel of volume
$V$ repel each other with intensity $a/V$. The one-component version
is a gas of particles with multi-particle interactions of
Curie-Weiss type. Its thermodynamic behavior is obtained by
integrating out the coordinates of one of the components of the
two-component version. In the grand canonical setting, a rigorous
theory of phase transitions in this model is developed and
discussed. In particular, for both versions thermodynamic phases and
phase diagrams are explicitly constructed  and the equations of
state are obtained and analyzed.

\end{abstract}



\maketitle

\section{Introduction}

The rigorous theory of thermal equilibrium of continuum particle
systems has got much more modest results than its counterpart
dealing with lattices, graphs, etc. There exist only few `realistic'
models in which the existence of a liquid-vapor phase transition was
mathematically proved. Among them there is the model introduced in
\cite{WR} by B. Widom and J. S. Rowlinson in which the potential
energy of $n$ point particles located at $x_1, \dots, x_n \in
\mathds{R}^d$ is set to be $\theta[ W(x_1 , \dots , x_n) - n]$,
where $\theta>0$ is a parameter and $W$ is the volume of the area
$\cup_{i=1}^n B(x_i)$ covered by the balls of unite volume centered
at these particles. The thermodynamics of this model is in a sense
equivalent to that of a two-component system with binary
interactions in which the interaction between unlike particles is a
hard-core repulsion and is zero otherwise. In \cite{R}, D. Ruelle
proved that the two-component system in two or more dimensions
undergoes a phase transition of first order. Later on, the rigorous
theory of this model was extended in \cite{Kot}, see also \cite{GHM}
for a review. However, these results give a little for understanding
the details of the phenomenon. No rigorous results are available on
the behavior at the phase-transition threshold. The very existence
of such a threshold remains unknown. At the same time, for a number
of lattice models the mean field approach allows for understanding
phase transitions in the corresponding models with `realistic'
interactions, see \cite{Biskup}. It is then quite natural to develop
the mean field theory of phase transitions also in continuum
systems. For the Widom-Rowlinson model, the first attempt to do this
was undertaken already in \cite[Sect. VII]{WR}. Assuming that the
particles are distributed in a given vessel ``at random" the authors
heuristically deduced an equation of state \cite[eq. (7.4)]{WR},
which manifests a first order phase transition. One of the ways to
develop a mean field theory in a rigorous way is to use Curie-Weiss
interaction potentials, see \cite{Zag} and \cite[Sect. IV.4]{EN}.
The aim of this work is to perform a rigorous study of this kind of
an analog of the Widom-Rowlinson model with Curie-Weiss
interactions, which we introduce in Section \ref{1S} below.
Similarly to the original Widom-Rowlinson model, it has two
versions: (a) a two-component gas of point particles with binary
repulsion of unlike particles; (b) a one-component gas with
multi-particle interactions the states of which are obtained by
`integrating out' the coordinates of one of the components of the
two-component system. The phase diagrams and the thermodynamic
phases of these versions are described in Theorems \ref{1tm} and
\ref{2tm}, respectively.  Unlike to \cite{Zag,LP} we work in the
grand canonical ensemble approach and -- along with a traditional
purely thermodynamic description -- we explicitly construct
thermodynamic phases and show their multiplicity occurring for
certain values of the particle activities. The formulation of the
results is followed by their detailed discussion in the same Section
\ref{1S}. The validity of Theorem \ref{2tm} directly follows from
Theorem \ref{1tm}.  The proof of Theorem \ref{1tm} is performed in
Section \ref{2S}. As is usual for Curie-Weiss interactions, the
thermodynamic limit is achieved by calculating asymptotics of
certain integrals, cf \cite[Theorem IV.4.1]{EN}. Unlike to lattice
system where this is mostly done by directly applying Laplace's
method, here we have to overcome technical difficulties related to a
more complex dependence of the integrands on the `large parameter'.

\section{The Setup }
\label{1S}

In this work, $\mathds{N}$ and $\mathds{R}$ will stand for the sets
of natural and real numbers, respectively; also $\mathds{N}_0 :=
\mathds{N}\cup \{0\}$. For $d \in \mathds{N}$, by $\mathds{R}^d$ we
denote the Euclidean space of vectors $x=(x^1, \dots, x^d)$, $x^i\in
\mathds{R}$, equipped with the usual Lebesgue measure $dx$.

\subsection{The model}

States of thermal equilibrium of infinite systems of point particles
in $\mathds{R}^d$ are described as probability measures defined on
the space of locally finite configurations $\Gamma=\{ \gamma \subset
\mathds{R}^d: |\gamma\cap \Lambda|<\infty\}$, where $\Lambda$ is a
\emph{vessel} -- a bounded closed subset of $\mathds{R}^d$, and
$|\gamma\cap \Lambda|$ stands for the number of particles in the
intersection of $\gamma$ with $\Lambda$. If the particles do not
interact, the corresponding state is a Poisson measure $P_z$,
characterized by \emph{activity} $z=e^\mu$. The dimensionless
parameter $\mu\in \mathds{R}$ is supposed to include the reciprocal
temperature $\beta$. For a vessel $\Lambda$ of volume $V$ and $n\in
\mathds{N}_0$, let $\Gamma_{\Lambda,n}$ be the set of all
configurations $\gamma$ such that $|\gamma\cap \Lambda|=n$. Then
$P_z$  is completely characterized by its values on all such sets
$\Gamma_{\Lambda,n}$, given by the formula
\begin{equation}
  \label{1}
 P_z (\Gamma_{\Lambda , n}) = \frac{\left( z V\right)^n}{n!} \exp
 \left( - z V \right).
\end{equation}
In the probabilistic interpretation,  $P_z$ assigns the probability
given in the right-hand side of (\ref{1}) to the event: $\Lambda$
contains $n$ particles. Assume now that point particles of two
types, 0 and 1, are placed in the same space $\mathds{R}^d$. If they
do not interact, their state of thermal equilibrium is the Poisson
measure $P_{z_0,z_1} = P_{z_0}\otimes P_{z_1}$, according to which
the event $\Gamma_{\Lambda,n_0} \times \Gamma_{\Lambda , n_1}$:
$\Lambda$ contains $n_0$ particles of type 0 and $n_1$ particles of
type 1, has the probability
\begin{equation}
  \label{2}
P_{z_0,z_1} \left(\Gamma_{\Lambda,n_0} \times \Gamma_{\Lambda , n_1}
\right) = P_{z_0}\left(\Gamma_{\Lambda,n_0} \right) \cdot
P_{z_1}\left(\Gamma_{\Lambda,n_1} \right),
\end{equation}
where $P_{z_i}\left(\Gamma_{\Lambda,n_i} \right)$, $i=0,1$, are as
in (\ref{1}).

For interacting particles, phases are constructed as limits $\Lambda
\to \mathds{R}^d$ of local Gibbs measures $P_{z}^{\Lambda,\Phi}$
($P_{z_0,z_1}^{\Lambda,\Phi}$ for two-component systems) describing
the portion of the particles contained in the vessel $\Lambda$ and
interacting with each other with energy $\Phi$, see, e.g.,
\cite{GHM,Ruelle,GeM}. In this work, we introduce two models that --
like the Widom-Rowlinson model -- can be considered as two versions
of the same model. The first one is a two-component gas of point
particles in $\mathds{R}^d$. For a vessel $\Lambda \subset
\mathds{R}^d$ of volume $V$, unlike particles contained in $\Lambda$
repel each other with intensity $a/V>0$, whereas like particles do
not interact. Thus, the potential energy of the collection of $n_0$
particles of type 0 located at $x_1^{0}, \dots, x_{n_0}^{0} \in
\Lambda$ and of $n_1$ particles of type 1 located at $x_1^{1},
\dots, x_{n_1}^{1} \in \Lambda$ is
\begin{equation}
  \label{U1}
 \Phi_\Lambda (x_1^{0}, \dots, x_{n_0}^{0}; x_1^{1}, \dots, x_{n_1}^{1}) =
 \sum_{i=1}^{n_0} \sum_{j=1}^{n_1}\frac{a}{V} = \frac{a}{V}n_0 n_1, \qquad n_0 , n_1 \in \mathds{N}_0.
\end{equation}
The grand canonical partition function of this collection then is
\begin{eqnarray}
  \label{U2}
& & \Xi_\Lambda (a, \mu_0 , \mu_1)\\[.2cm] \nonumber &  = & \sum_{n_0, n_1=0}^\infty
 \frac{1}{n_0! n_1 !} \int_{\Lambda^{n_0}}
\int_{\Lambda^{n_1}} \exp\left( \mu_0 n_0 + \mu_1 n_1 -
\frac{a}{V}n_0 n_1\right)d x_1^{0} \cdots d
x_{n_0}^{0} d x_1^{1} \cdots d x_{n_1}^{1}  \\[.2cm]\nonumber &= & \sum_{n_0, n_1=0}^\infty
 \frac{V^{n_0 + n_1}}{n_0! n_1 !}  \exp\left(
\mu_0 n_0 + \mu_1 n_1 - \frac{a}{V}n_0 n_1\right).
\end{eqnarray}
Here the interaction parameter $a>0$ and the chemical potentials
$\mu_i\in \mathds{R}$, $i=1,2$, include the reciprocal temperature
$\beta$ and thus are dimensionless. The second our model is a
one-component system of point particles interacting as follows. For
a vessel $\Lambda$ of volume $V$, the potential energy of the
collection of $n$ particles located at $x_0, \dots , x_n\in \Lambda$
is set to be
\begin{equation}
  \label{2a}
\widehat{\Phi}_\Lambda (x_1 , \dots , x_n) = V \theta \left[ 1 -
\exp\left( - \frac{a}{V}
 n\right)\right], \qquad n\in \mathds{N}_0.
\end{equation}
Here $\theta>0$ is a parameter, similar to that in \cite{WR}
mentioned above. Then the corresponding grand canonical partition
function is
\begin{eqnarray}
  \label{U3}
 \widehat{\Xi}_\Lambda (a, \mu , \theta) &  = & \sum_{n=0}^\infty
 \frac{1}{n!} \int_{\Lambda^{n}} \exp\left( \mu n - V \theta \left[ 1 -
\exp\left( - \frac{a}{V}
 n\right)\right]
\right)d x_1 \cdots d
x_{n} \\[.2cm]\nonumber &= & \sum_{n=0}^\infty
 \frac{V^{n}}{n!}  \exp\left(
\mu n  - V \theta \left[ 1 - \exp\left( - \frac{a}{V}
 n\right)\right] \right) \\[.2cm]\nonumber &= & \exp\left(
- V\theta\right) \Xi_{\Lambda}(a,\mu,\ln \theta).
\end{eqnarray}
The latter equality can readily be derived by summing out in
(\ref{U2}) over $n_1$. The dependence of the pressure $p$ in the
two-component system (resp. $\widehat{p}$ in the one-component
system) on $a$ and $\mu_i$, $i=0,1$ (resp. on $a$, $\mu$ and
$\theta$) is then obtained in the thermodynamic limit
\begin{eqnarray}
  \label{4}
  p = p(a,\mu_0, \mu_1)& = & \lim_{V\to +\infty} \frac{1}{V} \ln
  \Xi_{\Lambda}(a,\mu_0,\mu_1),\\[.2cm] \nonumber
  \widehat{p} =  \widehat{p}(a,  \theta, \mu)& = & \lim_{V\to +\infty} \frac{1}{V} \ln
   \widehat{\Xi}_{\Lambda}(a,\theta,\mu),
\end{eqnarray}
which by the last line in (\ref{U3}) yields $\widehat{p}= p -
\theta$. Thus, the particle density $\varrho$ in the one-component
system and the density $\varrho_0$ of the particles of type 0 in the
two-component system are related to each other by
\begin{equation}
  \label{4x}
\varrho = \frac{\partial \widehat{p}}{\partial \mu} = \frac{\partial
p}{\partial \mu_0}\bigg{|}_{\mu_0=\mu, \ \mu_1 = \ln \theta} =
\varrho_0\bigg{|}_{\mu_0=\mu, \ \mu_1 = \ln \theta}.
\end{equation}

\subsection{The results}

In the sequel, the two-component model defined in (\ref{U1}) and
(\ref{U2}) is considered as the main object of the study, and the
description of the one-component model is then based on the use of
(\ref{U3}) and (\ref{4x}).

\subsubsection{The two-component model}
According to (\ref{U2}) the two-component model is characterized by
three thermodynamic variables: $a, \mu_0, \mu_1$. Hence the
corresponding phase space is
\begin{equation}
  \label{U4}
\mathcal{F}= \{ (a, \mu_0, \mu_1): a \geq 0, \mu_0 , \mu_1 \in
\mathds{R}\}.
\end{equation}
We then define its subsets
\begin{gather}
  \label{U5}
 \mathcal{M}  =  \{ (a, \mu, \mu) : a>0, \mu > 1 - \ln a\},
 \\[.2cm] \nonumber  \mathcal{C}  =  \{ (a, 1-\ln a, 1 - \ln a) :
 a>0\}, \quad \mathcal{R}= \mathcal{F}\setminus (\mathcal{C}\cup
\mathcal{M}).
\end{gather}
Their meaning -- which will be seen below -- is as follows:
$\mathcal{M}$ is the set of phase coexistence points,  $\mathcal{C}$
is the line of the critical points and $\mathcal{R}$ is the
single-phase domain. The division itself is called the \emph{phase
diagram} of the model. It turns out that it is related to the maxima
of the function
\begin{equation}
  \label{10}
  E(y) = f(a,\mu_0 +y) + f(a,\mu_1 - y) - \frac{y^2}{2a}, \qquad y
  \in \mathds{R},
\end{equation}
with $a$, $\mu_1$ and $\mu_2$ considered as parameters. Here
\begin{eqnarray}
  \label{11}
  f(a,x) = \frac{a}{2} \left[ u(a,x)\right]^2 + u(a,x) , \qquad x\in
  \mathds{R},
\end{eqnarray}
whereas $u $ is a special function that can be expressed through
Lambert's $W$-function \cite{W} as follows
\begin{equation}
  \label{8}
  u(a,x) = \frac{1}{a} W ( a e^x).
\end{equation}
For a fixed $a>0$, the function $\mathds{R}\ni x \mapsto u(a, x)$
can be obtained as the inverse to
\begin{equation}
  \label{8b}
(0, +\infty)\ni u \mapsto x(u) = a u + \ln u,
\end{equation}
by which one gets that
\begin{gather}
  \label{8a}
a u(a,x) \exp\left[a u(a,x)  \right] = a e^x,\\[.2cm] \nonumber
  u' (a, x):= \frac{\partial }{\partial x} u(a, x)  =  \frac{ u(a, x)}{1 + a u(a,
  x)}.
\end{gather}
The relationship between (\ref{10}) and (\ref{U5}) is established in
the following statement, proved in Sect. 3.1 and illustrated in Fig
\ref{F1} below.
\begin{proposition}
  \label{0lm}
The function $E$ is infinitely differentiable on $\mathds{R}$ and
each of its global maxima is also a local maximum. Hence, it
satisfies the equation
\begin{equation}
  \label{12}
 y = w(y):= a u(a, \mu_0 + y ) - a u(a, \mu_1-y), \qquad y\in \mathds{R}.
\end{equation}
Moreover, the sets defined in (\ref{U4}) and (\ref{U5}) have the
following properties:
\begin{itemize}
  \item[{\it (a)}] For each $(a,\mu_0, \mu_1)\in \mathcal{R}$ such that $\mu_0\geq \mu_1$ (resp. $\mu_0\leq \mu_1$), $E$
  has a unique non-degenerate global maximum at some $y_*\geq 0$
  (resp. $y_*\leq 0$). For each $a>0$ and $\mu_0 = \mu_1 = 1 -\ln
  a$, $E$
  has a unique degenerate global maximum at $y=0$.
  \item[{\it (b)}] For each $a>0$ and $\mu_0 = \mu_1 =\mu > 1 -\ln
  a$, $E$ has two equal maxima at $\pm\bar{y}(a,\mu)$ where
  $\bar{y}(a,\mu)>0$ is a unique solution of the equation
\begin{equation}
\label{6b} \psi (y ) := y + \frac{y}{e^y -1} - 1 + \ln \frac{y}{e^y
-1} = \mu - (1-\ln a) , \quad y>0.
\end{equation}
\end{itemize}
\end{proposition}
In this statement, by saying that $y_*\in \mathds{R}$ is a
non-degenerate (resp. degenerate) maximum of $E$ we mean that its
second derivative satisfies $E''(y_*)<0$ (resp. $E''(y_*)<0$).

In the next statement -- the main result of this work -- we describe
the thermodynamics of the model at $(a,\mu_0,\mu_1)$ belonging to
$\mathcal{R}$ and $\mathcal{M}$. The behavior at the critical points
will be studied in a separate work.
\begin{theorem}
  \label{1tm}
The phase diagram of the model defined in (\ref{U1}) and  (\ref{U2})
is such that the following holds:
\begin{itemize}
\item[(i)] For each $(a,\mu_0,\mu_1)\in\mathcal{R}$, there exists a unique phase $P_{\tilde{z}_0 ,
\tilde{z}_1}$ with activities
\begin{equation}
  \label{7}
\tilde{z}_0 = u(a,\mu_0 + y_*) , \quad \tilde{z}_1 = u(a,\mu_1 -
y_*),
\end{equation}
where $y_*$ is the point of a unique global maximum of $E$
corresponding to this $(a,\mu_0,\mu_1)$.
\item[(ii)]  For each
$(a,\mu_0, \mu_1) \in \mathcal{M}$ (i.e., for $\mu_0 = \mu_1 = \mu >
1- \ln a$), there exist two phases: $P_{\tilde{z}^{+},
\tilde{z}^{-}}$ and $P_{\tilde{z}^{-} , \tilde{z}^{+}}$. Here
\begin{equation}
  \label{9}
\tilde{z}^{\pm} = u(a, \mu \pm \bar{y}(a,\mu)),
\end{equation}
and $\bar{y}(a,\mu)>0$ is the unique solution of the equation in
(\ref{6b}).
\item[(iii)]
The pressure $p$ defined in (\ref{4}) has the following form
\begin{equation}
  \label{9a}
 p = p(a, \mu_0, \mu_1)=a \varrho_0 \varrho_1 + \varrho_0 + \varrho_1,
\end{equation}
where the densities satisfy $\varrho_i = \tilde{z}_i$, $i=0,1$,  cf.
(\ref{7}).
\end{itemize}
\end{theorem}
The proof of this theorem will be done in the next section. Let us
now make some related comments. For $\mu_0 =\mu_1 =\mu < 1 - \ln a$,
by Proposition \ref{0lm} $E$ has a unique maximum at $y_*=0$.
According to (\ref{7}) both components of the system are then in the
same state $P_z$ with $z= u(a,\mu)$, and the state of the
two-component system $P_{z,z}$ is symmetric with respect to the
interchange of the components. By claim (ii), for each $\mu>1 - \ln
a$, there exist two different phases at the same $(a,\mu,\mu)\in
\mathcal{M}$, that breaks the symmetry between the components. Hence
we have a \emph{symmetry breaking} phase transition for which the
solution $\bar{y}(a, \mu)$ can serve as an order parameter. For
small $y>0$, we have that $\psi(y) = y^2/24 + o(y^2)$. Therefore,
$$\bar{y}(a, \mu) = \sqrt{24(\mu-(1-\ln a))} + o(\mu-(1-\ln a))$$ for
small positive $\mu - (1-\ln a)$. At the same time, for fixed
$\mu_0+\mu_1 = 2 - 2 \ln a$ and $\eta = \mu_0 - \mu_1$, we have that
\begin{equation*}
\bar{y}(a, \mu) = 2 \left( 6 \eta\right)^{1/3} + o(\eta^{1/3}),
\end{equation*}
which yields the classical mean-field value of the corresponding
critical exponent, cf \cite[Table V.2, page 173]{EN}. Each
$\varrho_i$ that appears in (\ref{9a}) depends on both $\mu_0$,
$\mu_1$ and satisfies $\varrho_i = \frac{\partial p}{\partial
\mu_i}$, $ i=0,1$, cf. (\ref{4x}). Note that the densities also
satisfy
\begin{equation}
  \label{9d}
\varrho_0 = \exp\left(\mu_0 - a \varrho_1 \right), \quad \varrho_1 =
\exp\left(\mu_1 - a \varrho_0 \right).
\end{equation}
That is, due to the repulsion both densities are smaller than they
are in the free case $a=0$.

Let us turn now to the ground states which one obtains by passing to
the limit $a\to +\infty$. To this end, we consider $\varrho_i$,
$i=0,1$, as differentiable functions of $a$ defined in (\ref{9d}).
Let $\dot{\varrho}_i$, $i=0,1$, stand for the corresponding
$a$-derivatives. Differentiating both sides of each equality in
(\ref{9d}) after some calculations we get
\begin{equation}
  \label{9e}
\dot{\varrho}_0 - \dot{\varrho}_1 = \frac{a \varrho_0 \varrho_1}{1 -
a^2 \varrho_0 \varrho_1}\left( \varrho_0 - \varrho_1 \right).
\end{equation}
The denominator here is positive by the fact that $y_*$ used in
(\ref{7}) is the point of local maximum of $E$ given in (\ref{10}).
Indeed, by claim (iii), (\ref{7}), (\ref{8a}) and  (\ref{10})  we
have that
\begin{eqnarray*}
1 - a^2 \varrho_0  \varrho_1 & = & 1 - a^2
u(a,\mu_0+y_*)u(a,\mu_1-y_*)\\[.2cm] &= &
- a E''(y_*) \left[ 1 + a u(a, \mu_0+y_*) \right] \left[ 1 + a u(a,
\mu_1-y_*) \right]>0.
\end{eqnarray*}
If $\mu_0
> \mu_1$, then $\varrho_0
> \varrho_1$ for all $a>0$.
For assuming $\varrho_0= \varrho_1$ for some $a>0$, we get by
(\ref{9d}) that $e^{\mu_0}=e^{\mu_1}$ and hence  $\mu_0 = \mu_1$.
Thus, by (\ref{9e}) $\varrho_0 - \varrho_1$ is an increasing
function of $a$, which yields $\varrho_0 - \varrho_1 \geq (\varrho_0
- \varrho_1)|_{a=0} = e^{\mu_0} - e^{\mu_1}$. By (\ref{9d}) and the
latter estimate we then get
\begin{eqnarray}
  \label{9f}
 \varrho_1 & = & \varrho_0 \exp\left(- (\mu_0 -\mu_1) - a (\varrho_0 - \varrho_1)
 \right) \\[.2cm] \nonumber
& \leq & \varrho_0 \exp\left(- (\mu_0 -\mu_1) - a (e^{\mu_0} -
e^{\mu_1})
 \right).
\end{eqnarray}
Since $\varrho_0 \leq e^{\mu_0}$, see (\ref{9d}), by (\ref{9f}) we
obtain that $a\varrho_1 \to 0$, and hence $\varrho_1 \to 0$, as
$a\to +\infty$. At the  same time, $\varrho_0 \geq \varrho_1 +
(e^{\mu_0}-e^{\mu_1})$, which by (\ref{9d}) yields that $\varrho_0
\to e^{\mu_0}$ as $a\to +\infty$. By (\ref{7}) we thus conclude that
the model has two ground states: $P_{z_0,0}$ and $P_{0,z_1}$. In
each of them, there is only one free component.

To illustrate the results described above we present in Fig.
\ref{F1} the part of the phase diagram in the plane in $\mathcal{F}$
with fixed $a>0$.
\begin{figure}[h!]
\includegraphics[width=240pt]{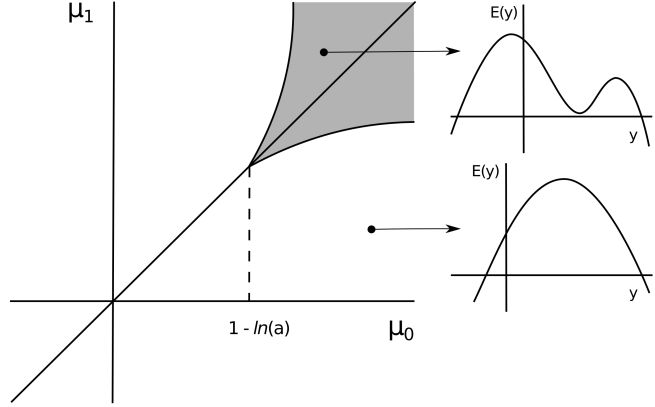}
\caption{Phase diagram at fixed $a$ }\label{F1}
\end{figure}
Points from the grayed area correspond to the existence of three
solutions of (\ref{12}), one of which is $y_*$. Note that $y_* >0$
for $\mu_0 > \mu_1$. At the boundaries of this area (symmetric under
 $\mu_0 \leftrightarrow \mu_1$), (\ref{12}) has only two
solutions. The upper branch of the boundary is described by the
equation
\begin{equation}
\label{xi}
\eta = \sqrt{\xi^2 -1} + \ln \left(\xi - \sqrt{\xi^2 -1}
\right),
\end{equation}
where $\xi = (\mu_0 + \mu_1 )/2 + \ln a$ and $\eta =
(\mu_1 - \mu_0)/2$. The lower branch is given also by (\ref{xi}) with the same $\xi$ and $\eta =( \mu_0 - \mu_1)/2$. For all points from the complement to the grayed
area, (\ref{12}) has only one solution. Note that $y_* =0$ for
$\mu_0 = \mu_1 < 1 - \ln a$.

\subsubsection{The one-component model}

In this case, the phase space is $\widehat{\mathcal{F}}=\{(a, \mu):
a\geq 0, \mu\in \mathds{R}\}$, cf (\ref{2a}) and (\ref{U4}).
\begin{theorem}
  \label{2tm}
The phase space $\widehat{\mathcal{F}}$ can be divided into disjoint
subsets $\widehat{\mathcal{R}}$, $\widehat{\mathcal{M}}$ and the
critical point $(e/\theta, \ln \theta)$. Here
$\widehat{\mathcal{M}}=\{ (a,\mu): a> e/\theta, \mu = \ln \theta\}$
is the phase coexistence line and $\widehat{\mathcal{R}}$ is the
single-phase domain. For each $(a, \mu)\in \widehat{\mathcal{R}}$,
there exists a unique phase $P_z$ with activity $z=u(a, \mu+y_*)$
where $y_*$ is the point of a unique global maximum of the function
defined in (\ref{10}) with $\mu_0 = \mu$ and $\mu_1 = \ln \theta$.
For each $(a, \mu)\in \widehat{\mathcal{M}}$, there exist two
phases: $P_{z^+}$ and $P_{z^-}$, where $z^{\pm}$  and the order parameter $\bar{y}$   are the same as  in (\ref{9}).
\end{theorem}
The proof of this theorem readily follows from Theorem \ref{1tm}.
The density of the particles in state $P_z$ is $\varrho = z = u(a,
\mu+y_*)$; it depends also on $\theta$. Moreover, by (\ref{9d}) we
have that $\varrho$ satisfies
\begin{equation}
  \label{U7}
  \varrho = \exp\left(\mu - \theta e^{- a \varrho} \right).
\end{equation}
It is an increasing and continuous function of $\mu$ whenever $a
\leq e/\theta$. For $a
> e/\theta$, $\varrho$ makes a jump at $\mu=\ln \theta$ with one-sided
limits $\lim_{\mu \to \ln \theta \pm 0}\varrho = u(a, \mu\pm
\bar{y}(a,\mu))$. That is, the system undergoes a first-order phase
transition with the increment of the density $\varDelta \varrho =
\bar{y}(a,\mu)/a$. The pressure defined in the second line of
(\ref{4}) is, cf (\ref{9a}),
\begin{equation}
  \label{9g}
 \widehat{p} = \widehat{p} (a, \mu)=  a \theta \varrho e^{- a \varrho} + \varrho - \theta
 \left(1 - e^{- a \varrho} \right),
\end{equation}
that can be obtained by (\ref{U7}) and the formula $\widehat{p}= p-
\theta$. Its dependence on $\mu\in \mathds{R}$ comes only from the
corresponding dependence of $\varrho$ just mentioned. In view of
this, for a fixed $a$, $\widehat{p}$ can also be considered as a
function of $\varrho$. This is typical for the corresponding works
employing the canonical partition function calculated for a fixed
number of particles $n$ in a vessel of volume $V$. Ten $\widehat{p}$
is obtained in the thermodynamic limit $n\to +\infty$, $V\to
+\infty$, taken in such a way that  $n/V \to \varrho$. In this case,
the density appears as an independent parameter of the theory, see,
e.g., \cite{LP,WR,Zag}. The drawback of this way is that, for $a
> e/\theta$, $\widehat{p}$ is a decreasing function of
$\varrho$ on a subinterval of $[z^{-}, z^{+}]$, which is impossible
from the physical point of view. The correct form of this dependence
can be deduced from the information on the dependence of $\varrho$
on $\mu$ discussed above. Namely, $\widehat{p}$ is given as in
(\ref{9g}) for $\varrho \leq z^{-}$ and $\varrho \geq z^{+}$. On the
interval $[z^{-}, z^{+}]$ it is constant, i.e.,
\begin{equation}
  \label{U8}
\widehat{p}\equiv \widehat{p}_* := a z^{+} z^{-} + z^{+}+ z^{-} -
\theta.
\end{equation}
In the canonical formalism, the horizontal part of the dependence of
$\widehat{p}$ on $\varrho$ may be obtained from the Maxwell rule, cf
\cite{LP}. To check whether this rule works in our case we have to
show that the following holds
\begin{equation}
  \label{U9a}
\int_{1/z^{+}}^{1/z^{-}} \left[ \widehat{p}(1/v) - \widehat{p}_*
\right] d v =0,
\end{equation}
which is equivalent to, cf (\ref{U8}),
\begin{equation}
  \label{U9}
- \int_{z^{-}}^{z^{+}} \left[ a \theta \varrho e^{-a \varrho} +
\varrho + \theta e^{-a \varrho}  \right] d \frac{1}{\varrho} = a
\left(  z^{+} - z^{-} \right)+ \frac{z^{+}}{z^{-}} -
\frac{z^{-}}{z^{+}}.
\end{equation}
We take into account that $a \theta \varrho e^{-a\varrho} d
\frac{1}{\varrho} = \frac{\theta}{\varrho} d e^{-a \varrho}$, and
then by integrating by parts we bring the left-hand side of
(\ref{U9}) to the following form
\begin{equation}
  \label{U10}
  {\rm LHS}(\ref{U9}) = \frac{\theta}{z^{-}} e^{-a z^{-}} - \frac{\theta}{z^{+}} e^{-a
  z^{+}} + \ln \frac{z^{+}}{z^{-}}.
\end{equation}
Since $z^{\pm} = u(a, \mu \pm \bar{y}(a, \mu))$, by the first line
in (\ref{8a}) it follows that
\[
\frac{\theta}{z^{\pm}} e^{-a z^{\pm}} = \theta e^{-\mu \mp
\bar{y}(a, \mu)} = e^{\mp \bar{y}(a, \mu)},
\]
where we have taken into account that $\mu = \ln \theta$ as $(a,
\mu)\in \widehat{\mathcal{M}}$. Likewise, we have that $z^{+} /
z^{-} = e^{\bar{y}(a,\mu)}$. On the other hand, by (\ref{12}) it
follows that $\bar{y}(a, \mu) = a (z^{+} - z^{-})$. We apply the
latter three facts in (\ref{U10}) and obtain
\begin{gather*}
{\rm LHS}(\ref{U9}) = e^{\bar{y}(a, \mu)} - e^{- \bar{y}(a, \mu)} +
\bar{y}(a, \mu) = \frac{z^{+}}{z^{-}} - \frac{z^{-}}{z^{+}} + a
(z^{+} - z^{-}) = {\rm RHS}(\ref{U9}),
\end{gather*}
which completes the proof of (\ref{U9a}).  Note that the equation of
state in (\ref{9g}) with $a=1$ formally coincides with that found
heuristically in \cite{WR}.

\section{Proving Theorem \ref{1tm}}
\label{2S}

We divide the proof into the following steps. First we prove
Proposition \ref{0lm} that relates the phase diagram (\ref{U5}) to
the properties of the function $E$. Thereafter, we relate $E$ with
the large $V$ asymptotic of $\ln \Xi_\Lambda/V$.

\subsection{The proof of Proposition \ref{0lm}}

By the very definition, see (\ref{10}) -- (\ref{8a}), it readily
follows that $E$ is an infinitely differentiable function. To prove
that it attains its global maxima not at infinity, let us show that
\begin{equation}
  \label{eE}
  E(y) \to - \infty, \qquad {\rm as} \ \ |y|\to +\infty.
\end{equation}
Since $E$ is symmetric with respect to the simultaneous interchange
$y \leftrightarrow -y$ and $\mu_0 \leftrightarrow \mu_1$, it is
enough to prove (\ref{eE}) for $y\to+\infty$. By (\ref{8b}) we have
\begin{eqnarray}
  \label{E1}
u(a,x)& < &  \frac{x}{a} - \ell (a,x), \qquad x>a, \\[.2cm]
\nonumber \ell (a,x) & = & \frac{1}{a} \ln \left(\frac{ x}{a} -
\frac{1}{a} \ln \frac{x}{a}\right).
\end{eqnarray}
Note that $\ell (a,x)>0$ for $x>a$ and $\ell(a,x) \to +\infty$ as
$x\to +\infty$. On the other hand, since $u(a,x)
>0$, for $x<0$ we have that $u(a,x)<e^x$ and hence $u(a,x)\to 0$ as
$x\to -\infty$. By (\ref{8a}) we get that $u$ is an increasing
function of $x$, which by (\ref{E1}) and (\ref{12}) yields that $y>
w(y)$ (resp. $y< w(y)$) for big enough $y$ (resp. $-y$). Thus
(\ref{12}) has at least one solution, say $y_{0}$. For $\mu_0
=\mu_1$, $w(0)=0$; hence, this solution gets positive for
$\mu_0>\mu_1$. By (\ref{10}) and (\ref{12}) we have that
\begin{equation}
  \label{E1a}
  E(y_0) = \frac{a}{2} u(a, \mu_0 +y_0) u(a, \mu_1 -y_0) + u(a, \mu_0
  +y_0) + u(a, \mu_1- y_0) >0,
\end{equation}
holding for all $y_0$ such that $w(y_0) = y_0$.

By (\ref{E1}), (\ref{10}) and (\ref{11}) for $y> \max\{ a-\mu_0;
\mu_1\}$ we obtain
\begin{eqnarray}
  \label{E2}
E(y) & < & \frac{a}{2}\left( \frac{\mu_0+y}{a} - \ell(a, \mu_0 +
y)\right)^2 + \frac{\mu_0+y}{a} - \ell(a, \mu_0 + y) \\[.2cm] \nonumber & + & \frac{a}{2}+
1 - \frac{y^2}{2a}  =  - A_1 (y) y - A_2 (y) \ell(a, \mu_0 + y)  +
\frac{\mu_0^2}{2a} + \frac{a}{2}+ 1,
\end{eqnarray}
with
\begin{gather}
  \label{E3}
A_1 (y) = \frac{1}{2}\ell(a, \mu_0 + y) - \frac{\mu_0+1}{a},
\\[.2cm] \nonumber
A_2 (y) = \frac{y}{2} - \frac{a}{2} \ell(a, \mu_0 + y) + \mu_0 + 1.
\end{gather}
Clearly, both these coefficients get positive for sufficiently big
$y$, which by (\ref{E2}) yields (\ref{eE}). This means that the
global maxima of $E$ are attained not at infinity and hence are also
local maxima, cf (\ref{E1a}). Therefore, the global maxima of this
function are to be found by solving the equation in (\ref{12}). Let
us first consider the case where $\mu_0 = \mu_1 = \mu$. Then $w$ is
an odd function and hence $y=0$ is a solution of (\ref{12}). By
(\ref{E1}) similarly as the estimate in (\ref{E2}) we obtain that
$w(y) < y$ for sufficiently large $y$, and hence $w(y) - y$ is
eventually negative. Obviously, the existence of positive solutions
of (\ref{12}) is determined by the slope of the curve $(w(y),y)$.
This means that we have to study the dependence of $w'(y)-1$ on $y$.
By means of (\ref{8a}) we get that
\begin{eqnarray}
  \label{E4}
  w'(y) -1 & = & \frac{c(y)}{[1 + a u(a, \mu +y)][1 + a u(a,
  \mu -y)]}, \\[.2cm] \nonumber c(y) & := & a^2 u(a, \mu +y) u(a, \mu -y)
  -1.
\end{eqnarray}
That is, the number of positive solutions of $w(y)=y$ coincides with
that of  $w'(y) =1$, and thus of $c(y)=0$. We apply (\ref{8a}) once
more and obtain
\begin{equation}
 \label{E5}
  c'(y) = \frac{a^3 u(a, \mu +y) u(a, \mu-y)}{[1 + a u(a, \mu +y)][1 + a u(a,
  \mu-y)]}\left[u(a, \mu-y) -
  u(a, \mu +y)\right].
\end{equation}
Since $u(a,\mu+y)$ is an increasing function of $y$, $c(y)$ has a
unique maximum at $y=0$. By the analysis made above regarding the
dependence of $u(a,x)$ on $x$ we conclude that $c(y) \to -1$ as
$y\to \pm\infty$. This and (\ref{E5}) imply that  $c$ has two real
zeros, say $\pm \hat{y}$, $\hat{y}>0$, whenever $c(0)
>0$. It has a single zero at $y=0$ if $c(0)
=0$. If $c(0)<0$, then $c(y)<0$ for all real $y$. In view of
(\ref{E4}),  we then have the following options: (i) $c(0)>0$, and
hence $w'(0)
>1$, which implies that $w(y) =y$ holds for $y=0$ and $y=\pm \bar{y}$, such
that $\bar{y} \geq \hat{y}$; (ii) $c(0)=0$, and hence $w'(0) =1$ and
$w'(y) <1$ for all $y>0$, which implies and $w(y) < y$ for all
$y>0$; (iii) $c(0) <0$, and hence $w(y) < y$ for all $y>0$. Let us
analyze these possibilities in terms of the parameters $\mu$ and
$a$. In case (ii), we have $u(a,\mu) =1/a$, which by (\ref{8a})
yields $\mu = \mu_c :=1-\ln a$ that determines a critical point, cf
(\ref{U5}). Since $c(0)$ is an increasing function of $\mu$, then
$c(0)>0$ implies that $\mu> \mu_c$ that corresponds to case (i).
Likewise, $\mu<\mu_c$ in case (iii). To relate this with $E$ we use
the fact that $E'' (y) = (w'(y) - 1)/a$, see (\ref{10}), (\ref{11})
and (\ref{12}). Thus, in case (i), $E$ has two equal non-degenerate
local (and also global) maxima at $\pm\bar{y}=\pm \bar{y}(a,\mu)$
and one local minimum at 0. In case (ii), $E$ has a degenerate
maximum at 0. In case (iii), this unique maximum gets
non-degenerate. This proves claim (b) of the statement, and the part
of (a) corresponding to the case of equal $\mu_i$. Let us show that,
for $y>0$, (\ref{12}) turns into (\ref{6b}). To simplify notations
by the end of this proof we set $v_{\pm}(y) = a u(a,\mu \pm y)$.
Then $v_{\pm}(y) \exp(v_{\pm}(y)) = \exp(\mu +\ln a \pm y)$, see
(\ref{8a}). We combine this with (\ref{12}) in the form $y =
v_{+}(y) -v_{-}(y)$ to obtain
\begin{equation}
  \label{E6}
 v_{+}(y) = \frac{ye^y}{e^y-1}, \quad v_{-}(y) = \frac{e^y}{e^y-1}.
\end{equation}
Then we rewrite (\ref{8b}) in the form
\[
\ln v_{-}(y) + v_{-}(y) + y -1 = \mu - (1-\ln a),
\]
that by (\ref{E6}) coincides with (\ref{6b}).

To complete the proof we have to consider the case of unequal
$\mu_i$. In view of the mentioned symmetry of $E$, it is enough to
consider the case $\mu_0 \geq \mu_1$. Set $\mu_1 = \mu$ and $\mu_0 =
\mu +\delta$, and then
\begin{equation}
  \label{E7}
w(\delta , y)= a u(a, \mu +\delta +y) - a u(a, \mu -y).
\end{equation}
By (\ref{E1}) we have that $w(\delta , y) < y$ for sufficiently
large $y$. At the same time, $w(\delta , 0) >0$ for $\delta >0$.
That is, (\ref{12}) has at least one positive solution, say $y_*$,
in this case. It is such that $E''(y_*) = (w'(y_*)-1)/a <0$; i.e.,
$E$ has a non-degenerate maximum at $y_*$. By standard arguments
based on the implicit function theorem we have that $y_*$ is a
continuous function of $\delta\geq0$ that tends to a nonnegative
solution of (\ref{12}) as $\delta \to 0^+$. Its $\delta$-derivative
$\dot{y}_*$ can be calculated from the equality $y_* =
w(\delta,y_*)$, which yields
\begin{equation}
  \label{E8}
\dot{y}_* = \frac{a u (a, \mu+\delta + y_*)} {\left[ 1 - w'(\delta ,
y_*)\right]\left[ 1+ a u (a, \mu+\delta + y_*)\right]} >0.
\end{equation}
That is, for $\mu\leq \mu_c$ and $\delta >0$, $y_*>0$ is the only
maximum point of $E$, and $y_* \to 0$ as $\delta \to 0^+$. For $\mu
>\mu_c$, by the positivity in (\ref{E8}) we have that $y_* >
\bar{y}(a,\mu)$. In this case, we have two more solutions of $w(0,y)
=y$. By the $\delta$-continuity of the solutions of $y =
w(\delta,y)$ it should have two more solutions, say $y_1$ and $y_0$,
close to $-\bar{y}(a,\mu)$ and zero, respectively, for small enough
$\delta>0$. Their derivatives have the form as in (\ref{E8}) with
$y_*$ replaced by the corresponding $y_j$. Since $w'(\delta , y_0)$
is close to $w'(0,0)$, then $\dot{y}_0 <0$, and hence $y_0 <0$. At
the same time, $\dot{y}_1 >0$ for the same reason. That is, these
two solutions move towards each other as $\delta$ increases. Let us
compare the values of $E$ at $y_*$ and $y_1$. For $\delta =0$, we
have that $E(y_1)=E(y_*)$. The $\delta$-derivative $\dot{E}(y)$ can
be calculated from (\ref{10}), which yields
\begin{equation*}
\dot{E}(y) = u (a, \mu + \delta + y) + \dot{y}\left[w(\delta , y)-y
\right]/a = u (a, \mu + \delta + y).
\end{equation*}
Here we have taken into account that $w(\delta , y)=y$ for $y=y_1,
y_*$. Thus, $\dot{E}(y_*) >\dot{E}(y_1)$ since $y_* > 0 > y_1$ and
$u$ is an increasing function of $y$. This means that $y_*$ is the
point of non-degenerate local and global maximum of $E$.

As follows from this proof, $y_1 < y_0<0$ for $\mu>\mu_c$ and small
$\delta>0$. Let us fix $\mu>\mu_c$ and find $\delta>0$ and $y<0$
such that $y_1=y_0 =y$. Note that (\ref{12}) has two solutions in
this case: this $y$ and $y_*>0$. Clearly such $\delta$ and $y$ are
to be found from the equation $w'(\delta, y)=1$. Similarly as above,
set $v_{+} (y) = a u(a, \mu +\delta +y)$, $v_{-} (y) = a u(a, \mu
-y)$. Then $w'(\delta, y)=1$ by (\ref{E4}) yields $v_{+}(y) v_{-}
(y) =1$. By (\ref{8a}) we have
\[
v_{+}(y) v_{-} (y) \exp\left[ v_{+}(y)+ v_{-} (y) \right] =
\exp\left( 2 \mu + \delta + 2 \ln a \right),
\]
by which we get
\begin{equation}
  \label{E10}
v_{+}(y) + v_{-} (y) = 2 \xi := 2 \mu + \delta + 2 \ln a.
\end{equation}
Since $y<0$, we have that $v_{+}(y) < v_{-} (y)$. Keeping this in
mind we solve (\ref{E10}) and $v_{+}(y) v_{-} (y) =1$, which yields
\begin{equation}
  \label{E11}
v_{\pm}(y) = \xi \mp \sqrt{\xi^2 -1}, \quad y = - 2 \sqrt{\xi^2-1}.
\end{equation}
By (\ref{8a}) we have
\[
\frac{v_{+}(y)}{v_{-}(y)} \exp\left( {v_{+}(y)} - {v_{-}(y)}\right)
= \exp\left(2\eta + 2y  \right), \quad \eta := \delta/2.
\]
Now we use here (\ref{E11}) and arrive at (\ref{xi}).

\subsection{Thermodynamics in a fixed vessel}

In equilibrium statistical mechanics, the great canonical ensemble
is determined by the family of local Gibbs measures indexed by all
possible vessels $\Lambda$, see \cite[Chapter 4]{Ruelle}. Such
measures are in turn uniquely determined by their correlation
functions. For a given vessel $\Lambda$ and $x^0_1 , \dots ,
x^0_{n_0}, x^1_1, \dots , x^1_{n_1}\in \Lambda$, the correlation
function $k^{(n_0, n_1)}_\Lambda (x_1^0, \dots , x^0_{n_0}; x_1^1,
\dots , x^1_{n_1})$ is defined as the density (with respect to the
Lebesgue measure) of the probability distribution of the particles
of both types in $\Lambda$. If the potential energy $\Phi_\Lambda$
is given, then
\begin{gather}
  \label{K1}
k^{(n_0, n_1)}_\Lambda (x_1^0, \dots , x^0_{n_0}; x_1^1, \dots ,
x^1_{n_1}) = \frac{1}{\Xi_\Lambda} \sum_{m_0, m_1=0}^\infty
\frac{z_0^{n_0+m_0} z_1^{n_1+m_1}}{m_0! m_1 !} \\[.cm] \nonumber
\times \int_{\Lambda^{m_0}} \int_{\Lambda^{m_1}} \exp\left( -
\Phi_\Lambda (x_1^0, \dots , x^0_{n_0}, y_1^0, \dots , y^0_{m_0} ;
x_1^1, \dots , x^1_{n_1}, y_1^1, \dots , y^1_{m_1}) \right) \\[.cm] \nonumber
\times dy_1^0 \cdots ,d y^0_{m_0} \cdot d y_1^1 \cdots  d y^1_{m_1},
\end{gather}
where $z_0$, $z_1$ and $\Xi_\Lambda$ are the corresponding
activities and the partition function, respectively. The correlation
functions of the states of the whole infinite system can be obtained
in the limit $\Lambda \to \mathds{R}^d$. For the Poissonian state
defined in (\ref{1}) and (\ref{2}), we have that
\begin{equation}
  \label{K2}
k^{(n_0, n_1)} (x_1^0, \dots , x^0_{n_0}; x_1^1, \dots , x^1_{n_1})
= z_0^{n_0} z_1^{n_1}, \qquad n_0, n_1 \in\mathds{N}_0.
\end{equation}
Now for $\Phi_\Lambda$ as in (\ref{U1}) and fixed $\mu_0$, $\mu_1$,
we thus have, cf (\ref{U2}) and (\ref{K1}),
\begin{eqnarray}
  \label{K3}
& & k^{(n_0, n_1)}_\Lambda (x_1^0, \dots , x^0_{n_0}; x_1^1, \dots ,
x^1_{n_1}) =  \exp\left(\mu_0 n_0 +\mu_1 n_1 - \frac{a}{V} n_0 n_1
\right) \\[.cm] \nonumber & & \qquad
\times \frac{1}{\Xi_\Lambda (a,\mu_0, \mu_1)} \sum_{m_0,
m_1=0}^\infty \frac{V^{m_0+m_1}}{m_0! m_1 !} \exp\bigg{(}\left[\mu_0
- \frac{a}{V}n_1\right]m_0 \\[.cm] \nonumber & & \qquad + \left[\mu_1 - \frac{a}{V}m_0\right]m_1
-
\frac{a}{V} m_0 m_1 \bigg{)} \\[.cm] \nonumber & & \qquad =
\exp\left(\mu_0 n_0 +\mu_1 n_1 - \frac{a}{V} n_0 n_1 \right)
\frac{\Xi_\Lambda (a,\mu_0 - a n_1/V, \mu_1 - a n_0/V )}{\Xi_\Lambda
(a,\mu_0, \mu_1)}
\end{eqnarray}
Set
\begin{eqnarray}
  \label{K4}
 F_\Lambda (a , \mu_0 , \mu_1) & = & \frac{1}{V}\ln \Xi_\Lambda (a , \mu_0 ,
 \mu_1), \\[.2cm] \nonumber F^{(i)}_\Lambda (a , \mu_0 , \mu_1) & = &
 \frac{\partial}{\partial \mu_i}F_\Lambda (a , \mu_0 , \mu_1), \quad
 i=0,1,
\end{eqnarray}
and rewrite (\ref{K3}) in the following form
\begin{equation}
  \label{K5}
k^{(n_0, n_1)}_\Lambda (x_1^0, \dots , x^0_{n_0}; x_1^1, \dots ,
x^1_{n_1}) =  \exp\left(\tilde{\mu}^\Lambda_0 n_0
+\tilde{\mu}^\Lambda_1 n_1 - \frac{a}{V} n_0 n_1 \right),
\end{equation}
with
\begin{gather}
  \label{K6}
\tilde{\mu}^\Lambda_0 = \mu_0 - a \int_0^1 F^{(1)}_\Lambda \left(a,
\mu_0, \mu_1 - \frac{a}{V} n_0 t \right) dt, \\[.2cm] \nonumber
\tilde{\mu}^\Lambda_1 = \mu_1 - a \int_0^1 F^{(0)}_\Lambda \left(a,
\mu_0 -\frac{a}{V} n_1 t , \mu_1 - \frac{a}{V} n_0  \right) dt.
\end{gather}
Thus, we have to show that
\begin{eqnarray}
  \label{K7}
& & \tilde{\mu}^\Lambda_0 \to \ln \tilde{z}_0 = \mu_0 - a u(a, \mu_1
- y_*), \\[.2cm] \nonumber
& & \tilde{\mu}^\Lambda_1 \to \ln \tilde{z}_1 = \mu_1 - a u(a, \mu_0
+ y_*),
\end{eqnarray}
as $V\to +\infty$, see (\ref{7}),  (\ref{9d}) and (\ref{K2}). This
means that we have to obtain the large $V$ asymptotic of the
functions defined in (\ref{K4}). To this end by means of the
identity $$- \frac{a}{V}n_0 n_1 = - \frac{a}{2V}n^2_0 -
\frac{a}{2V}n^2_1 + \frac{a}{2V}(n_0 - n_1)^2,$$ and then by the
standard Gaussian formula
$$\exp\left(\frac{ b^2}{2V}\right) =
\sqrt{\frac{V}{2\pi}}\int_{-\infty}^{+\infty} \exp\left( b y -
\frac{V y^2}{2}\right) d y,
$$
we rewrite (\ref{U2}) and (\ref{K4}) in the form
\begin{equation}
  \label{20}
\Xi_\Lambda (a, \mu_0 , \mu_1) = \exp\left( V F_\Lambda (a, \mu_0 ,
\mu_1\right)= \sqrt{\frac{V}{2\pi}} \int_{-\infty}^{+\infty}
\exp\left(V E_V (y) \right) d y,
\end{equation}
with
\begin{equation}
  \label{21}
 E_V(y) = f_V (a, \mu_0 +y) + f_V(a, \mu_1-y) - \frac{y^2}{2a}.
\end{equation}
Here $f_V$ is defined by the following formula
\begin{equation}
  \label{22}
\exp\left( V f_V (a, x)\right) = \sum_{n=0}^\infty \frac{V^n}{n!}
\exp\left( x n - \frac{a}{2V} n^2 \right),
\end{equation}
and thus is an infinitely differentiable function of $x\in
\mathds{R}$ for each fixed $a>0$ and $V>0$. Then so is $ E_V$ as a
function of $y\in \mathds{R}$. Moreover, taking the
$\mu_i$-derivatives of both sides of (\ref{20}) we obtain
\begin{gather}
  \label{K8}
F^{(0)}_\Lambda (a,\mu_0, \mu_1) = \frac{\int_{-\infty}^{+\infty}
u_V (a , \mu_0 +y) \exp\left(V E_V (y)\right) d
y}{\int_{-\infty}^{+\infty}  \exp\left(V E_V (y)\right) d y},
\\[.2cm] \nonumber
F^{(1)}_\Lambda (a,\mu_0, \mu_1) = \frac{\int_{-\infty}^{+\infty}
u_V (a , \mu_1 -y) \exp\left(V E_V (y)\right) d
y}{\int_{-\infty}^{+\infty}  \exp\left(V E_V (y)\right) d y},
\end{gather}
where, cf (\ref{22}),
\begin{gather}
  \label{K9}
u_V (a, x)  = \frac{\partial}{\partial x} f_V (a , x) =
\frac{\langle n \rangle_V}{V} = \frac{1}{V}\sum_{n=1}^\infty n \pi_V
(x,n), \\[.2cm] \nonumber
\pi_V (x,n) = \frac{V^n}{n!} \exp\left( x n - \frac{a}{2V} n^2
\right) \bigg{/}\sum_{n=0}^\infty \frac{V^n}{n!} \exp\left( x n -
\frac{a}{2V} n^2 \right).
\end{gather}
To find the large $V$ asymptotic of the right-hand sides of
(\ref{20}) and (\ref{K8}) we employ a more advanced version of
Laplace's method as $E_V$ depends on $V$. Namely, we will use
\cite[Theorem 2.2, Chapter II]{Fed} which we present here in the
form adapted to the context.
\begin{proposition}
  \label{1pn}
Assume  that, for all big enough $V$, the function defined in
(\ref{21}) has a unique non-degenerate global maximum at some
$y_{*,V}\in \mathds{R}$, so that its second $y$-derivative satisfies
$E''_V(y_{*,V})< 0$. Assume also that there exists a function $V
\mapsto \alpha_V>0$ such that $\alpha_V \to +\infty$ and
\begin{equation}
  \label{K10}
\Delta_V:= \frac{\alpha_V}{\sqrt{V |E''_V (y_{*,V})|}} \to 0 , \quad
{\rm as} \ \ V \to +\infty.
\end{equation}
Set $U_V = [y_{*,V} - \Delta_V,y_{*,V} + \Delta_V]$ and let
$\phi_V(y)$ be constant or either of $u_V (a,\mu_0+ y)$, $u_V
(a,\mu_1 -y)$, cf (\ref{K8}). Then in the limit of large $V$, it
follows that
\begin{eqnarray}
  \label{K11}
& & \int_{U_V} \phi_V(y) \exp\left( V E_V (y)\right)
d y \\[.2cm] \nonumber  & & \qquad \qquad = \sqrt{\frac{2\pi}{V|E''_V (y_{*,V})|}} \phi_V (y_{*,V})
\exp\left( V E_V (y_{*,V})\right)\left[1 + o(1) \right].
\end{eqnarray}
\end{proposition}

\subsection{Preparatory statements}

In this  subsection, we obtain a number of results by means of which
we then apply Proposition \ref{1pn} in (\ref{K8}). We begin by
obtaining some bounds on the first two $x$-derivatives of the
function defined in (\ref{K9}) which we denote by $u'_V (a,x)$ and
$u''_V (a,x)$.
\begin{lemma}
  \label{1lm}
For each $x\in \mathds{R}$ and $V>0$, the following holds
\begin{gather}
  \label{K12}
  0 \leq u'_V (a, x) \leq u_V (a, x).
\end{gather}
\end{lemma}
\begin{proof}
By taking the $x$-derivative in (\ref{K9}) we get
\begin{equation}
  \label{K15}
 u'_V (a, x) = \langle \left(n-\langle n \rangle_V\right)^2
 \rangle_V {/}V\geq 0,
\end{equation}
which proves the lower bound stated in (\ref{K12}). On the other
hand, by taking the $x$-derivative of both sides of (\ref{22}) we
obtain that $u_V$ satisfies, cf (\ref{8a}),
\begin{eqnarray}
  \label{23}
u_V (a,x) \exp\left( a\int_0^1 u_V \left( a, x- \frac{a}{V} t\right)
dt\right)= \exp\left( x -  \frac{a}{2V} \right).
\end{eqnarray}
Now we differentiate both sides of (\ref{23}) and obtain
\begin{gather}
  \label{K14}
 u'_V (a, x) + a u_V (a,x) \int_0^1 u'_V \left( x - \frac{a}{V} t\right)
 d  t = u_V (a,x).
\end{gather}
In view of (\ref{K15}), the second summand here is positive which
yields the upper bound in (\ref{K12}).
\end{proof}
Recall that $u(a,x)$ is defined in (\ref{8}).
\begin{corollary}
  \label{1co}
For each $x\in \mathds{R}$ and $V>0$, the following holds
\begin{equation}
  \label{K16}
  u\left(a, x- \frac{a}{2V} \right) \leq u_V (a,x) \leq u\left(a, x +\frac{a}{2V} \right)
  ,
\end{equation}
and hence
\begin{equation}
  \label{K16a}
  \left\vert u_V (a, x) - u(a,x)\right\vert \leq \frac{1}{2V}.
\end{equation}
\end{corollary}
\begin{proof}
By (\ref{K12}) $u_V\left(a, x \right)$ is an increasing function of
$x$, which by (\ref{23}) yields
\begin{gather}
  \label{K17}
u_V \left(a, x- \frac{a}{V} \right) \exp\left(a u_V \left(a, x-
\frac{a}{V} \right)  \right) \leq \exp\left(x - \frac{a}{2V}
\right), \\[.2cm] \nonumber
u_V \left(a, x \right) \exp\left(a u_V \left(a, x \right)  \right)
\geq \exp\left(x - \frac{a}{2V} \right).
\end{gather}
On the other hand, by (\ref{8a}) we have that
\[
\exp\left(x - \frac{a}{2V} \right) = u\left(a, x - \frac{a}{2V}
\right) \exp\left( a  u\left(a, x - \frac{a}{2V} \right) \right).
\]
Since the function $u \mapsto u e^{au}$ is increasing, the first
line in (\ref{K17}) implies that
\[
u_V \left(a, x- \frac{a}{V} \right) \leq u\left(a, x - \frac{a}{2V}
\right),
\]
which yields the upper bound in (\ref{K16}). The lower bound is
obtained from the second line in (\ref{K17}) analogously. Then the
estimate in (\ref{K16a}) follows by these bounds and the fact that
$u'(a,x) \leq 1/a$, see (\ref{8a}).
\end{proof}
\begin{lemma}
  \label{2lm}
For each $a>0$, there exists a continuous function $x\mapsto h_a(x)
>0$ such that, for all $V>0$, the following holds
\begin{equation}
 \label{K13}
\left\vert u''_V (a , x)\right\vert \leq h_a (x).
\end{equation}
\end{lemma}
\begin{proof}
Similarly as in (\ref{K15}) we get
\begin{equation}
  \label{K15a}
 u''_V (a, x) = \langle \left(n-\langle n \rangle_V\right)^3
 \rangle_V {/}V.
\end{equation}
However, unlike to (\ref{K15}) we have no information on the sign of
this derivative. The idea of proving (\ref{K13}) is to split $u''_V
(a, x)$ into two parts, one of which is positive and the other one
is controllable. Then the first part can be controlled similarly as
in Lemma \ref{1lm}. To this end we use a certain property of the
probability distribution defined in the second line of (\ref{K9}).
Namely, we want to find its modes: all those $n_*$ that satisfy the
conditions
\begin{equation}
  \label{K15b}
 \frac{\pi_V (a,n_* \pm 1)}{\pi_V(a,n_*)} \leq 1.
\end{equation}
By taking `minus' in (\ref{K15b}) we obtain from (\ref{K9}) and
(\ref{8a}) that
\begin{equation}
  \label{K15c}
  \frac{n_*}{V} \exp\left( a \frac{n_*}{V}\right) \leq \exp\left( x+
  \frac{a}{2V}\right) = u\left( x+
  \frac{a}{2V}\right) \exp\left( a u \left( x+
  \frac{a}{2V}\right)\right).
\end{equation}
Likewise, by taking `plus' in (\ref{K15b}) we get
\begin{equation}
  \label{K15d}
  \frac{n_* +1}{V} \exp\left( a \frac{n_*+1}{V}\right) \geq \exp\left( x+
  \frac{a}{2V}\right).
\end{equation}
Since the function $u \mapsto ue^{au}$ is increasing on
$(0,+\infty)$, the inequalities and equalities in (\ref{K15c}) and
(\ref{K15d}) imply that
\begin{equation}
  \label{K15e}
V u \left( x+
  \frac{a}{2V}\right) -1 \leq n_* \leq  V u \left( x+
  \frac{a}{2V}\right),
\end{equation}
and hence the probability distribution defined in the second line of
(\ref{K9}) is unimodal. By (\ref{K9}) we have that $\langle n
\rangle_V = V u_V (a, x)$. Then we use the estimates in (\ref{K16})
and obtain from (\ref{K15e}) the following
\begin{gather*}
 n_* - \langle n \rangle_V \leq V \left[ u \left( x+
  \frac{a}{2V}\right) - u \left( x-
  \frac{a}{2V}\right)\right], \\[.2cm] = \frac{a}{2} \int_{-1}^1 u' \left(a, x + \frac{a}{2V}t \right) dt \leq 1 , \nonumber
\end{gather*}
where we used the estimate $a u' (a,x) \leq 1$ which readily follows
from the second line in (\ref{8a}). On the other hand, also by the
estimates in (\ref{K16}) we get that $\langle n \rangle_V  - n_*
\leq 1 $, which finally yields
\begin{equation}
  \label{L22}
\left\vert n_* - \langle n \rangle_V \right\vert\leq 1,
\end{equation}
holding for all $x \in \mathds{R}$ and $V>0$. Now keeping in mind
(\ref{K15a}) we write
\begin{eqnarray}
  \label{24}
\langle \left( n - \langle n \rangle_V \right)^3\rangle_V & = &
\langle \left( n - n_* + \frac{1}{2}\right)^3\rangle_V + 3 \left(
n_* - \frac{1}{2} - \langle n \rangle_V \right)
\\[.2cm] \nonumber & \times & \left[ \langle \left( n - \langle n \rangle_V\right)^2\rangle_V + \left( n_* - \frac{1}{2} - \langle n \rangle_V \right)^2\right]
- 2 \left( n_* - \frac{1}{2} - \langle n \rangle_V \right)^3  \\[.2cm] \nonumber & = &   \langle \left( n - n_* +      \frac{1}{2}\right)^3\rangle_V +
3 \left( n_* - \frac{1}{2} - \langle n \rangle_V \right)  \langle \left( n - \langle n \rangle_V\right)^2\rangle_V   \\[.2cm] \nonumber
& + & \left( n_* - \frac{1}{2} - \langle n \rangle_V \right)^3.
\end{eqnarray}
Set
\begin{equation}
\label{L23} g_V (a,x) = \frac{1}{V} \langle \left( n - n_* +
\frac{1}{2}\right)^3\rangle_V.
\end{equation}
Then by (\ref{24}) and  (\ref{L22}) we have that
\begin{gather}
  \label{L24}
 \left\vert u''_V (a,x) - g_V (a,x) \right\vert \leq \frac{9}{2}
 \left[ u'_V (a , x) + \frac{3}{4V}\right] \\[.2cm] \nonumber \leq \frac{9}{2}
 \left[ u \left(a , x+ \frac{a}{2V_0}\right) + \frac{3}{4V_0}\right] =: \chi(x),
\end{gather}
where we assume that $V\geq V_0$ for some fixed $V_0$ and use the
upper bounds in (\ref{K12}) and (\ref{K16}). To estimate $g_V$ we
write
\begin{eqnarray}
  \label{27}
\langle \left( n - n_* +      \frac{1}{2}\right)^3\rangle_V & = &
\sum_{n=0}^\infty \left( n - n_* +      \frac{1}{2}\right)^3 \pi_V (x,n) \\[.2cm] \nonumber &
\geq & \sum_{m=0}^{n_*-1} \left(m +\frac{1}{2} \right)\left(\pi_V
(x,n_* +m) - \pi_V (x,n_*-m-1) \right).
\end{eqnarray}
By the second line in (\ref{K9}) we have
\begin{eqnarray*}
& & \frac{\pi_V (x,n_* +m)}{\pi_V(x,n_* -m -1)} \\[.2cm]
& & \qquad \qquad = \frac{V^{2m+1}\exp\left((2m+1)x -
\frac{a}{2V}\left[ (n_* +m)^2 - (n_*-m-1)^2\right]\right)}
{(n_*+m) (n_*+m-1)\cdots n_* \cdots (n_*-m+1) (n_*-m)} \\[.2cm] & & \qquad \qquad  =
\left[ \left( \frac{V}{n_*}\right) \exp\left(x - a \frac{n_*}{V} +
\frac{a}{2V} \right)\right]^{2m+1} \bigg{/} \prod_{k=1}^m \left(1 -
\left( \frac{k}{n_*} \right)^2\right) \geq 1,
\end{eqnarray*}
where the latter estimate follows by the inequality in (\ref{K15c}).
Then by (\ref{27}) and (\ref{L23}) we conclude that, for all $x\in
\mathds{R}$ and $V\geq V_0$,
\begin{equation}
  \label{28}
g_V (x) \geq 0.
\end{equation}
By (\ref{K14}) we get
\begin{equation}
  \label{L25}
a \int_0^1 u'_V \left( a, x - \frac{a}{V} t \right) dt = 1 -
\frac{u'_V (a, x)}{u_V (a,x)}.
\end{equation}
Now we take the $x$-derivative of both sides of (\ref{K14}), use
(\ref{L25}) and obtain
\begin{equation*}
u''_V ( a, x) + a u_V (a,x) \int_0^1 u''_V\left(a, x - \frac{a}{V} t
\right) dt = \frac{\left[u'_V (a, x)\right]^2}{u_V (a,x)} \leq
u\left(a, x + \frac{a}{2V}\right),
\end{equation*}
where we also use the upper bounds in (\ref{K12}) and (\ref{K16}).
We write here $u''_V = g_V + ( u''_V - g_V)$, use the estimate
obtained in (\ref{L24}) and the positivity in (\ref{28}). This
yields
\begin{gather}
\label{L27} g_V (a,x) \leq g_V (a,x) + a u_V (a,x) \int_0^1 g_V
\left(a,x-
\frac{a}{V}t\right) dt \\[.2cm] \nonumber \leq u\left(a, x +
\frac{a}{2V_0}\right) + \chi(x) + a u\left(a, x +
\frac{a}{2V_0}\right) \int_0^1 \chi \left( x - \frac{a}{V_0} t
\right) dt \\[.2cm] \nonumber
\leq \chi(x) + u\left(a, x + \frac{a}{2V_0}\right) \left( 1 + a
\chi(x) \right),
\end{gather}
where we also use that $\chi$ is an increasing function, see
(\ref{L24}) and (\ref{8a}). Thus, by the latter and (\ref{L24}) we
conclude that the estimate stated in (\ref{K13}) holds true with
$h_a= \chi +{\rm RHS}(\ref{L27})$.
\end{proof}
\begin{corollary}
  \label{2co}
In the limit $V\to +\infty$, we have that $u'_V \to u'$ given in
(\ref{8a}), point-wise in $a$ and uniformly on compact subsets of
$\mathds{R}$ in $x$.
\end{corollary}
\begin{proof}
We integrate by parts in (\ref{K14}) and obtain therefrom that
\begin{equation*}
u'_V (a, x) = \frac{u_V (a,x)}{1+ a u_V (a,x)} \left[ 1 +
\frac{a}{V} \int_0^1 (1-t) u_V'' \left( a, x- \frac{a}{V} t\right)
dt\right].
\end{equation*}
Then the proof follows by (\ref{K16a}), (\ref{K13}) and the fact
that $u'(a,x) = u(a,x) /( 1 + au(x,a))$, see (\ref{8a}).
\end{proof}
By (\ref{22}) we have that $f_V(a,x) \leq e^x$ and hence $f_V(a,x)
\to 0$ as $x\to -\infty$. By (\ref{K9}) this yields
\begin{equation*}
  f_V (a,x) = \int_{-\infty}^x u_V (a, y) dy,
\end{equation*}
which by (\ref{K16}) leads to
\begin{equation}
  \label{L30}
f \left(a, x - \frac{a}{2V}\right) \leq f_V (a, x) \leq f \left(a, x
+ \frac{a}{2V}\right).
\end{equation}
Then for $V\geq V_0$, we have that
\begin{equation}
  \label{L31}
  \left\vert f_V (a,x) - f(a,x) \right\vert \leq  \frac{a}{2V} u\left(a,
  x+ \frac{a}{2V_0}
  \right).
\end{equation}
\begin{lemma}
  \label{4lm}
For each $a>0$, we have that $E_V\to E$ as $V\to +\infty$ uniformly
on compact subsets of $\mathds{R}$. We also have that
\begin{eqnarray}
  \label{L32}
& & E'_V (y) \to E'(y)= u(a,\mu_0 +y) - u(a, \mu_1 -y) - \frac{y}{a}, \\[.2cm] \nonumber
& & E''_V (y) \to E''(y)= \frac{u(a,\mu_0 +y)}{1 + a u(a,\mu_0 +y)}
+ \frac{u(a,\mu_1 -y)}{1 + a u(a,\mu_1 -y)} - \frac{1}{a},
\end{eqnarray}
where the convergence of the first (resp. second) derivatives is
uniform (resp. uniform on compact subsets) in $y$.
\end{lemma}
\begin{proof}
The convergence $E_V \to E$ follows by (\ref{L31}) and the fact that
$f'(a,x) = u(a,x)$ is bounded in $x$ on compact subsets of
$\mathds{R}$. The uniform in $y$ convergence $E'_V (y) \to E'(y)$
follows by (\ref{K16a}); the convergence of the second derivatives
follows by Corollary \ref{2co}.
\end{proof}
\begin{lemma}
  \label{5lm}
Assume that $(a,\mu_0, \mu_1) \in \mathcal{R}$, and hence the
function $E$ defined in (\ref{10}) has a unique non-degenerate
global maximum at the corresponding $y_*\in \mathds{R}$, see
Proposition \ref{0lm}. Then there exist $V_0>0$, $\varepsilon>0$ and
$y_{\pm}$ such that $y_{-} < y_{+}$, $y_* \in [y_{-} , y_{+}]$ and
for all $V>V_0$ the following holds:
\begin{itemize}
  \item[(i)] the function $E_V$ defined in (\ref{21}) has
also a unique global maximum at some $y_{*,V}\in [y_{-} , y_{+}]$;
\item[(ii)] $- E''_V (y) \geq \varepsilon$ for all $y \in [y_{-} ,
y_{+}]$;
\item[(iii)] $y_{*,V} \to y_*$ as $V\to +\infty$.
\end{itemize}
\end{lemma}
\begin{proof}
We begin by recalling that the assumptions imposed on $E$ imply that
$E''(y_*) < 0$. Set, cf (\ref{E7}) and (\ref{L32}),
\begin{equation}
  \label{L32a}
 w (y) = a u(a,\mu_0 +y) - au(a,\mu_1-y) = a E'(y) +y.
\end{equation}
Then, cf (\ref{E4}),
\begin{eqnarray}
  \label{L32b}
  w'(y) - 1 & = & \frac{c(y)}{[1 + a u(a, \mu_0 +y)][1 + a u(a,
  \mu_1-y)]}, \\[.2cm] \nonumber
c(y) & := & a^2 u(a,\mu_0 +y) u(a,\mu_1-y)-1.
\end{eqnarray}
By (\ref{8a}) we get
\begin{equation}
  \label{L32c}
 c'(y) = \frac{a^3 u(a,\mu_0 +y) u(a,\mu_1-y)}{[1 + a u(a, \mu_0 +y)][1 + a u(a,
  \mu_1-)]}\left[u(a,\mu_1-y) - u(a,\mu_0+y) \right],
\end{equation}
and hence $c'(y)=0$ at $y = - (\mu_0-\mu_1)/2$, where $c$ has
maximum. since $E''(y_*)<0$, we have that $w'(y_*) <1$ (see
(\ref{L32a})), and hence
\begin{equation}
  \label{L33}
a^2 u(a,\mu_0 +y_*) u(a,\mu_1-y_*)-1 <0.
\end{equation}
Set,
\begin{equation}
  \label{L33a}
  w_V (y) = a u_V (a, \mu_0 +y) - a u_V (a,\mu_1 - y).
\end{equation}
By (\ref{K16a}) (resp. Corollary \ref{2co}) it follows that $w_V \to
w$ (resp. $w'_V \to w'$) as $V\to +\infty$, point-wise in $a$ and
uniformly in $y$ (resp. uniformly in $y$ on compact subsets of
$\mathds{R}$).

As above, we assume that $\mu_0\geq \mu_1$. For $\mu_0 < \mu_c = 1 -
\ln a$, we have that  $w'(y) < 1$ for all $y\in \mathds{R}$, see
Fig. \ref{F1}, and hence $w (y) < y$ for all $y>y_*$, and $w (y) >
y$ for all $y<y_*$. Fix any $y_{\pm}$ such that $y_{-} <y_{+}$ and
$y_* \in [y_{-}, y_{+}]$, then pick positive $V_0$ and $\varepsilon$
such that, for all $V>V_0$, the following holds: (a) $w_V (y_{+}) <
y_{+}$, $w_V (y_{-}) > y_{-}$; (b)  $1 - w_V'(y) \geq \varepsilon a$
for all $y \in [y_{-}, y_{+}]$. This is possible in view of the
convergence just mentioned. By (a) we then have that there exists a
unique $y_{*,V}\in [y_{-}, y_{+}]$ such that $w_V (y_{*,V})=
y_{*,V}$ which is an extremum point of $E_V$. In view of the
convergence stated in Lemma \ref{4lm}, this is the point of
non-degenerate global maximum. By (b) we have that (ii) holds true.
Thus, it remains to prove the validity of claim (iii) in this case.
By the very definition of $y_*$ and $y_{*,V}$ we have that $y_* -
y_{*,V} = w(y_{*} ) - w_V (y_{*,V})$. Then
\begin{gather}
  \label{L34}
|y_* - y_{*,V}| \leq |w(y_*) - w_V(y_*)| + |w_V (y_*) - w_V(
y_{*,V})| \\[.2cm] \nonumber \leq \frac{a}{V} + (1-a \varepsilon) |y_* -
y_{*,V}|,
\end{gather}
which yields $|y_* - y_{*,V}| < 1/V\varepsilon$. Here we have taken
in to account (\ref{K16a}) and the fact that $y_*, y_{*,V}\in
[y_{-}, y_{+}]$. Let us now consider the case $\mu_0> \mu_c$. Set
$\delta = \mu_0 - \mu_1$ and $\xi (\delta)= (\mu_0+\mu_1)/2 + \ln a
= \mu_0 - \mu_c + 1 + \delta/2$,  $\eta (\delta) = \delta/2$. Let
$\delta_*>0$ be defined by the condition that $\xi (\delta_*)$ and
$\eta (\delta_*)$ satisfy (\ref{xi}). For $\delta \geq \delta_*$,
$E$ has a single non-degenerate global maximum, and the proof of the
lemma is the same as in the case of $\mu_0 < \mu_c$. Thus, we ought
to consider the case $\delta \in (0,\delta_*)$ where $E$ has two
local maxima, say at $y_1$ and $y_*>y_1$, and one local minimum at
$y_0 \in [y_1, y_*]$, see Fig. \ref{F1} and Lemma \ref{0lm}. For
$\mu_0 > \mu_1$, $y_*$ is the point of non-degenerate global maximum
of $E$. Since $c(y)$ defined in (\ref{L32b}) is continuous, by
(\ref{L33}) it follows that there exists $y_{-} \in(0, y_*)$ such
that $c(y_{-})<0$ and $w(y_{-}) > y_{-}$. Note that $c(-\delta
/2)>0$ for $\mu_0 > \mu_c$. Set $2\varkappa = w(y_{-}) - y_{-}$ and
then pick $y_{+} > y_{*}$ such that $y_{+} - w(y_{+}) \geq
2\varkappa$, which is possible in view of (\ref{E1}). By
(\ref{L32c}) we have that $c'(y) <0$ for $y>0$; hence,
$[y_{-},+\infty) \ni y\mapsto c(y)/(1 + a u(a, \mu_1 - y))$ is a
decreasing function. Thus, for all $y\in [y_{-}, y_{+}]$, by
(\ref{L32b}) we have that
\begin{equation}
  \label{L35}
  1 - w(y) \geq  2 a \varepsilon,
\end{equation}
with
\[
\varepsilon := - \frac{c(y_{-})}{2 a [1 + a u(a, \mu_0 +y_{+})][1 +
a u(a, \mu_1 -y_{-})]}.
\]
Then by the convergence of $w_V$ and $w'_V$ discussed above, see
(\ref{L33a}),  and (\ref{L35} we conclude that there exists $V_0$
such that, for all $V>V_0$, the following holds: (a) $w_V (y_{-}) -
y_{-} \geq \varkappa$, and $y_{+} - w_V (y_{+}) \geq \varkappa$; (b)
$1 -w_V(y) \geq a \varepsilon$ holding for all $y\in [y_{-},
y_{+}]$. Thereafter, the proof of all the three claims of the lemma
follows in the same way as in the case of $\mu_0 < \mu_c$.
\end{proof}
\begin{remark}
  \label{1rk}
For $\mu_0\geq \mu_1$, we have that $- E''_V(y) >0$ for all $y>
y_{*,V}$ and $V>V_0$. This can be seen from the fact
 that $w'(y) -1$ vanishes just once for $y\geq y_{-}$ and from the
 convergence $w_V' \to w'$.
\end{remark}
Finally, we study the thermodynamic limit for $(a,\mu_0, \mu_1)\in
\mathcal{M}$, where $\mu_0=\mu_1=\mu>\mu_c=1-\ln a$, see (\ref{U5}),
and thus $E_V$ is an even function, see (\ref{21}). The proof of the
next statement follows by the same arguments that were used in the
proof of Lemma \ref{5lm}, case $\mu_0 >\mu_c$.
\begin{lemma}
  \label{6lm}
Assume that $(a,\mu, \mu) \in \mathcal{M}$, and hence $E$ has two
equal non-degenerate maxima at $\pm\bar{y}(a,\mu)$. Then there exist
$V_0>0$, $\varepsilon>0$ and $\upsilon$ such that for all $V>V_0$
the following holds:
\begin{itemize}
  \item[(i)] there exists $y_{*,V}\in [\bar{y}(a,\mu) - \upsilon, \bar{y}(a,\mu) + \upsilon]$ such that $E_V(y_{*,V}) \geq E_V(y)$ for all $y\geq 0$;
\item[(ii)] $- E''_V (y) \geq \varepsilon$ for all $y \in [\bar{y}(a,\mu) - \upsilon, \bar{y}(a,\mu) + \upsilon]$;
\item[(iii)] $y_{*,V} \to \bar{y}(a,\mu)$ as $V\to +\infty$.
\end{itemize}
\end{lemma}

\subsection{The proof of Theorem \ref{1tm}}

Basically, to complete the proof we have to show that: (a) the
phases are as stated in claims (i) and (ii); (b) the following
holds, cf (\ref{4}), (\ref{9a}) and (\ref{20}),
\begin{equation}
  \label{L36}
 \lim_{V\to +\infty} F_\Lambda (a, \mu_0 , \mu_1) = a u(a, \mu_0 +
 y_*) u(a, \mu_1 -
 y_*) + u(a, \mu_0 +
 y_*) + u(a, \mu_1 -
 y_*).
\end{equation}
The proof of (a) will be done by showing the convergence stated in
(\ref{K7}), which by (\ref{K6}) also amounts to studying the
asymptotic properties of the integrals in (\ref{20}) and (\ref{K8}).
To this end we use Proposition \ref{1pn}, cf (\ref{K11}). First we
consider the case $(a,\mu_0, \mu_1) \in \mathcal{R}$, see Lemma
\ref{5lm}.
\begin{lemma}
  \label{7lm}
Assume that $(a,\mu_0, \mu_1) \in \mathcal{R}$ and let $y_{*,V}$ be
as in Lemma \ref{5lm}. Then in the limit $V\to +\infty$ we have that
\begin{eqnarray}
  \label{KL11}
& & \int_{-\infty}^{+\infty} \phi_V(y) \exp\left( V E_V (y)\right)
d y \\[.2cm] \nonumber  & & \qquad \qquad = \sqrt{\frac{2\pi}{V|E''_V (y_{*,V})|}} \phi_V (y_{*,V})
\exp\left( V E_V (y_{*,V})\right)\left[1 + o(1) \right].
\end{eqnarray}
\end{lemma}
\begin{proof}
Let $V_0$ and $\varepsilon$ be as in Lemma \ref{5lm}. Then $-
E''(y_{*,V}) \geq \varepsilon$ and hence $\Delta_V$ defined in
(\ref{K10}) with $\alpha_V =V^{1/4}$ tends to zero. Let $I_V$ stand
for the left-hand side of (\ref{K11}) with such $\Delta_V$. Let also
$I^{-}_V$ and $I_V^{+}$ stand for the integrals over $(-\infty,
y_V^{-}]$ and $[y_V^{+}, +\infty)$, $y_V^{\pm}:= y_{*,V} \pm
\Delta_V$, respectively, so that ${\rm LHS}(\ref{KL11})= I_V +
I_V^{+} + I_V^{-}$. In view of Lemma \ref{5lm}, the proof of
(\ref{KL11}) will be done by showing that
\begin{equation}
  \label{L36w}
I^{\pm}_V \exp\left(-V E_V (y_{*,V}) \right) \to 0, \qquad {\rm as}
\ \ V \to +\infty.
\end{equation}
As above, we set $\mu_0 \geq \mu_1$, and hence $y_{*,V} \geq 0$. Let
$y_{\pm}$ be in Lemma \ref{5lm}. Since $\Delta_V \to 0$, we have
that $y_{+}> y_V^{+}= y_{*,V} +\Delta_V$ and $y_{-}< y_V^{-}=
y_{*,V} -\Delta_V$, holding  for big enough $V$. By (\ref{E1}) and
(\ref{E3}) both $A_1$ and $A_2$ in the estimate in (\ref{E2}) are
increasing functions of $y$. Let $b_{+}> y_{+}$ (resp. $b_{-}<
\min\{0;y_{-}\}$) and positive $C^{+}_0, C_1^{+}$,
$C^{+}_0<C_1^{+}b_{+}$ (resp. $C_0^{-}, C_1^{-}$, $C_0^{-}< -
C_1^{-}b_{-}$) be such that the following version of (\ref{E2})
holds
\begin{equation}
  \label{L36a}
  E(y) < \left\{ \begin{array}{ll}  C_0^{+} - C_1^{+} y, \qquad &{\rm for} \ \ y\geq
b_{+};\\[.2cm]  C_0^{-} + C_1^{-} y, \qquad &{\rm for} \ \ y\leq
b_{-}. \end{array} \right.
\end{equation}
By (\ref{L30}) we have that $E_V$ also satisfies (\ref{L36a}) for
all $V>V_0$. Since $y_{*,V}$ is neither in $[y_{+}, b_{+}]$ nor in
$[ b_{-}, y_{-}]$, there exists $\epsilon>0$ such that
\begin{equation}
  \label{L36q}
  E_V (y_{*,V}) - \sup_{y\in [ b_{-}, y_{-}] }
  E_V(y) \geq \epsilon, \quad  E_V (y_{*,V}) - \sup_{y\in [y_{+}, b_{+}]}
  E_V(y) \geq \epsilon.
\end{equation}
Let $b_{\pm}$ be as just described. For all assumed choices of
$\phi_V$, one can pick positive $c_0$, $c_1$ and $c_2$ such that:
\begin{equation}
  \label{L37q}
  \phi_V (y) \leq \left\{ \begin{array}{ll}
c_0 + c_1 y, \quad &{\rm for} \ \ y\geq b_{+},\\[.2cm]
c_2, \quad &{\rm for} \ \ y\in [b_{-},b_{+}]\\[.2cm]
c_0 - c_1 y, \quad &{\rm for} \ \ y\leq b_{-}.
  \end{array}\right.
\end{equation}
Set
\begin{eqnarray}
  \label{L37}
I^{+}_V & = & I^{+,0}_V + I^{+,1}_V + I^{+,2}_V\\[.2cm] \nonumber & := &
\int_{y_V^{+}}^{y_{+}}\phi_V(y) \exp\left( V E_V (y)\right) d y +
\int_{y_{+}}^{b_{+}}\phi_V(y) \exp\left( V E_V (y)\right) d y \\[.2cm] \nonumber &
+ & \int_{b_{+}}^{+\infty}\phi_V(y) \exp\left( V E_V (y)\right) d y.
\end{eqnarray}
By (\ref{L36a}) and (\ref{L37q}) we obtain
\begin{gather}
  \label{L38}
I^{+,2}_V \leq \exp\left( V C_0^{+}\right) \int_{b_{+}}^{\infty}
(c_0 +
c_1 y) \exp\left( - V C_1^{+} y \right) dy \\[.2cm] \nonumber =
\frac{1}{C_1^{+} V} \exp\left( V[C_0^{+} - C_1^{+} b_{+}]\right)
\left(c_0 + c_1
b_{+} + \frac{c_1}{C_1^{+} V} \right)\\[.2cm] \nonumber \leq \frac{1}{C_1^{+} V} \exp\left( V E_V (y_{*,V}) \right) \left(c_0 +
c_1 b_{+} + \frac{c_1}{C_1^{+} V} \right).
\end{gather}
Here we have taken into account that $C_0^{+} - C_1^{+} b_{+} \leq
0$ and $E_V (y_{*,V})>0$, see (\ref{E1a}) and (\ref{L30}).

To estimate $I^{+,1}_V $ we set $\epsilon_V^{+} = \sup_{y\in [y_{+},
b_{+}]} E_V (y)$ and use the corresponding estimate from
(\ref{L37q}). By (\ref{L36q}) this yields
\begin{gather}
  \label{L38z}
I^{+,1}_V \leq  c_2 e^{V\epsilon_V^{+}} (b_{+} - y_{+}) \leq c_2
(b_{+} - y_{+}) \exp\left( - V \epsilon + VE_V (y_{*,V})\right).
\end{gather}
To estimate $I^{+,0}_V $ we use the fact that $- E''_V(y) \geq
\varepsilon$ for all $y\in [y^{+}_V, y_{+}]$. That is, $h_V(y):= - V
E_V$ is convex and increasing on this interval. Set $\tau = h_V(y)$, and hence $dy = d \tau/ h'_V (y)$. Then, cf
(\ref{L37q}),
\begin{gather}
  \label{L39}
I^{+,0}_V = \int_{h_V(y_V^{+})}^{h(y_{+})} \frac{\phi_V (y)}{h'_V
(y)} e^{-\tau} dt  \leq \frac{c_2}{V\varepsilon}
\int_{h_V(y_V^{+})}^{h_V(y_{+})} e^{-\tau} d \tau \\[.2cm] \nonumber \leq
\frac{c_2}{V\varepsilon} \exp\left( V E_V (y_V^{+})\right) \leq
\frac{c_2}{V\varepsilon} \exp\left( V E_V (y_{*,V})\right).
\end{gather}
Now we use (\ref{L38}), (\ref{L38z}) and (\ref{L39}) in (\ref{L37})
and obtain that (\ref{L36}) holds true for $I^{+}_V$. Write
\begin{eqnarray*}
I^{-}_V & = & I^{-,0}_V + I^{-,1}_V + I^{-,2}_V\\[.2cm] \nonumber & := &
\int_{y_{-}}^{y_V^{-}}\phi_V(y) \exp\left( V E_V (y)\right) d y +
\int_{b_{-}}^{y_{-}}\phi_V(y) \exp\left( V E_V (y)\right) d y \\[.2cm] \nonumber &
+ & \int_{-\infty}^{b_{-}}\phi_V(y) \exp\left( V E_V (y)\right) d y,
\end{eqnarray*}
where $y_{-}$ is the same as in Lemma \ref{5lm} and $b_{-}$ as in
(\ref{L36a}). Then we proceed exactly as in (\ref{L38}),
(\ref{L38z}) and (\ref{L39}) to show that (\ref{L36w}) holds true
also for $I^{-}_V$.
\end{proof}
Now we consider the case where $(a,\mu, \mu)\in \mathcal{M}$ and
thus $E_V$ is an even function, cf (\ref{21}). In particular,
$E_V(y_{*,V}) = E_V(-y_{*,V})$.
\begin{lemma}
  \label{8lm}
Assume that $(a,\mu, \mu)\in \mathcal{M}$ and let  $y_{*,V}$ be as
in Lemma \ref{6lm}. Then in the limit $V\to +\infty$ we have that
\begin{eqnarray}
  \label{KL12}
& & \int_{-\infty}^{+\infty} \phi_V(y) \exp\left( V E_V (y)\right)
d y \\[.2cm] \nonumber  & & \qquad \qquad = \sqrt{\frac{2\pi}{V|E''_V (y_{*,V})|}} \left[ \phi_V
(-y_{*,V})+ \phi_V (y_{*,V})\right] \exp\left( V E_V
(y_{*,V})\right)\left[1 + o(1) \right].
\end{eqnarray}
\end{lemma}
\begin{proof}
Set
\begin{eqnarray}
  \label{L41q}
  \varphi_V (y) & = & \phi_V (-y) + \phi_V (y),\\[.2cm]
I_V & = & \int_{0}^{+\infty} \varphi_V(y) \exp\left( V E_V
(y)\right) d y. \nonumber
\end{eqnarray}
Thus, we have to show that
\begin{equation}
  \label{L42}
I_V = \sqrt{\frac{2\pi}{V|E''_V (y_{*,V})|}}  \varphi_V ( y_{*,V})
\exp\left( V E_V (y_{*,V})\right)\left[1 + o(1) \right].
\end{equation}
Let $V_0$, $\varepsilon$ and $\upsilon$ be as in Lemma \ref{6lm} and
$\Delta_V= V^{-1/4}/\sqrt{ \varepsilon}$, cf (\ref{K10}). Set
$y^{\pm}_V = y_{*,V} \pm \Delta_V$, $y_{\pm} = \bar{y}(a,
 \mu) \pm \upsilon$ and assume that $V$ is big enough so that $y_{-}
 < y^{-}_V$ and $y_{+}
 < y^{+}_V$.  Let $C_0^{+}$, $C_1^{+}$ and $b_{+}$ be such
that the first line in (\ref{L36a}) holds true. We also assume that
both estimates in (\ref{L36q}) hold where $b_{-}$ is set to be zero.
Finally, by (\ref{L37q}) we have that
\begin{equation}
  \label{L43}
\varphi_V (y) \leq \left\{ \begin{array}{ll} c_2, \quad &{\rm for}
\ \ y\in [0, b_{+}]; \\[.2cm] c_0 + c_1 y , \quad &{\rm for}
\ \ y > b_{+},
\end{array}\right.
\end{equation}
holding for all $V>V_0$. Then we split
 $I_V$ into six summands, i.e., write $I_V=\sum_{j=1}^6
 I_{j,V}$. In estimating these summands we mainly follow the way
 elaborated in proving Lemma \ref{7lm}. Namely, cf (\ref{L38z}),
 \begin{gather}
\label{L44}
 I_{1,V} := \int_0^{y_{-}} \varphi_V(y) \exp\left( V E_V (y)\right) d
y \leq c_2 y_{-} \exp\left( - V \epsilon + VE_V (y_{*,V})\right).
\end{gather}
Next, set $\tau = h_V(y):= - V E_V (y)$, cf (\ref{L39}),
\begin{gather}
  \label{L45}
 I_{2,V} := \int_{y_{-}}^{y^{-}_V} \varphi_V(y) \exp\left( V E_V (y)\right) d
y  = \int_{h_V(y_{-})}^{h_V(y_V^{-})} \frac{\varphi_V (y)}{h_V'(y)}
e^{-\tau} d \tau \\[.2cm] \nonumber \leq \frac{c_2}{V
\varepsilon} \int_{h_V(y_{-})}^{h_V(y_V^{-})}  e^{-\tau} d \tau \leq
\frac{c_2}{V \varepsilon} \exp\left(VE_V (y_V^{-}) \right) \leq
\frac{c_2}{V \varepsilon} \exp\left(VE_V (y_{*,V}) \right).
\end{gather}
The next integral is estimated by means of Proposition \ref{1pn}.
That is,
\begin{eqnarray}
  \label{L46}
 I_{3,V} & := & \int_{y_V^{-}}^{y_V^{+}} \varphi_V(y) \exp\left( V E_V (y)\right) d
y \\[.2cm] \nonumber & = & \sqrt{\frac{2\pi}{V|E''_V (y_{*,V})|}}  \varphi_V (y_{*,V})
\exp\left( V E_V (y_{*,V})\right)\left[1 + o(1) \right].
\end{eqnarray}
The next one is estimated pretty similar to (\ref{L45})
\begin{eqnarray}
  \label{L47}
 I_{4,V} & := & \int_{y_V^{+}}^{y_{+}} \varphi_V(y) \exp\left( V E_V (y)\right) d
y \\[.2cm] \nonumber &\leq & \frac{c_2}{V \varepsilon} \exp\left(VE_V (y_V^{+})
\right) \leq \frac{c_2}{V \varepsilon} \exp\left(VE_V (y_{*,V})
\right).
\end{eqnarray}
The next integral in turn is estimated similarly  as in (\ref{L44})
 \begin{gather}
\label{L44a}
 I_{5,V} := \int_{y_{+}}^{b_{+}} \varphi_V(y) \exp\left( V E_V (y)\right) d
y \leq c_2 (b_{+}- y_{+}) \exp\left( - V \epsilon + VE_V
(y_{*,V})\right).
\end{gather}
Finally, cf (\ref{L38}) and (\ref{L43}),
\begin{eqnarray}
  \label{L48}
I_{6,V} & := & \int_{g_{+}}^{+\infty} \varphi_V(y) \exp\left( V E_V
(y)\right) d y \\[.2cm] \nonumber &\leq &\exp\left(V C_0^{+}\right)
\int_{g_{+}}^{+\infty} \left(c_0 + c_1 y \right) \exp\left(
- V C_1^{+} y\right) dy \\[.2cm] \nonumber & \leq & \frac{1}{ C_1^{+}V} \exp\left(VE_V
(y_{*,V})\right) \left[ c_0 + c_1 b_{+} +
\frac{c_1}{C_1^{+}V}\right].
\end{eqnarray}
Now by (\ref{L44}), (\ref{L45}), (\ref{L46}), (\ref{L47}),
(\ref{L44a}) and (\ref{L48}) we conclude that (\ref{L42}) holds
true.
\end{proof} \vskip.1cm \noindent {\it Proof of Theorem \ref{1tm}.}
First we consider the case $(a,\mu_0, \mu_1)\in \mathcal{R}$. Apply
Lemma \ref{7lm} in (\ref{K8}) with $\phi_V (y) = u_V(a, \mu_0+y)$ in
the numerator and $\phi_V (y) \equiv 1$ in the denominator. This
yields
\begin{equation}
  \label{L49}
F^{(0}_\Lambda (a, \mu_0 , \mu_1) = u_V (a, \mu_0 + y_{*,V})\left[ 1
+ o(1)\right].
\end{equation}
On the other hand, by (\ref{K16a}) and (\ref{L34}) we obtain
\[
u_V (a, \mu_0 + y_{*,V}) = u (a , \mu_0 +y_{*}) + o(1)
\]
Since $y_{*}$ is a continuously differentiable function of $\mu_0$
and $\mu_1$, cf (\ref{E8}), we have that
\[
F^{(0}_\Lambda \left(a, \mu_0 - \frac{a}{V}n_1 t,   \mu_1 -
\frac{a}{V}n_0\right) = u (a , \mu_0 +y_{*}) + o(1),
\]
uniformly in $t\in[0,1]$. We use this in (\ref{K6}) and obtain that
the second line in (\ref{K7}) holds true. The proof of the first
line follows analogously. This proves claim (i) of the theorem. Let
us now turn to the case $(a,\mu, \mu)\in \mathcal{M}$. By Lemma
\ref{8lm} we obtain, cf (\ref{L49}) and (\ref{L41q}),
\begin{equation}
  \label{L50}
F^{(0}_\Lambda (a, \mu_0 , \mu_1) = \frac{1}{2}\left[ u_V (a, \mu-
y_{*,V}) + u_V (a, \mu + y_{*,V}) \right](1+o(1)).
\end{equation}
On the other hand, by (\ref{U2}) and then by (\ref{K4}) it follows
that
\begin{eqnarray*}
F^{(0}_\Lambda (a, \mu_0 , \mu_1) & = & \varrho_{0,\Lambda}:=
\frac{1}{\Xi_\Lambda (a,
\mu_0,\mu_1)}\\[.2cm] \nonumber &\times &\sum_{n_0, n_1=0}^\infty \left(\frac{n_0}{V}\right)
 \frac{V^{n_0 + n_1}}{n_0! n_1 !}  \exp\left(
\mu_0 n_0 + \mu_1 n_1 - \frac{a}{V}n_0 n_1\right).
\end{eqnarray*}
That is, $F^{(0}_\Lambda (a, \mu_0 , \mu_1)$ is the density of the
particles of type 0 in the local state corresponding to the
interaction energy (\ref{U1}) (determined by $a$) and chemical
potentials $\mu_0$ and $\mu_1$. By (\ref{L50}) we have
\begin{equation}
  \label{L52}
  \lim_{\Lambda \to \mathds{R}^d} \varrho_{0,\Lambda} = \frac{1}{2}\left[ u (a, \mu-
\bar{y} (a,\mu)) + u (a, \mu + \bar{y} (a,\mu)) \right].
\end{equation}
Likewise,
\[
\lim_{\Lambda \to \mathds{R}^d} \varrho_{1,\Lambda} = {\rm
RHS}(\ref{L52}),
\]
and
\begin{eqnarray*}
\lim_{\Lambda \to \mathds{R}^d} F^{(01)}_\Lambda (a, \mu_0 , \mu_1)
& = & \frac{1}{2}\bigg{(} u (a, \mu- \bar{y} (a,\mu)) u (a, \mu+\bar{y} (a,\mu))  \\[.2cm] \nonumber & + &
 u (a, \mu + \bar{y} (a,\mu))u (a, \mu - \bar{y}
(a,\mu)) \bigg{)},
\end{eqnarray*}
where
\begin{eqnarray*}
F^{(01)}_\Lambda (a, \mu_0 , \mu_1) & = & \frac{1}{\Xi_\Lambda (a,
\mu_0,\mu_1)}\\[.2cm] \nonumber &\times &\sum_{n_0, n_1=0}^\infty
\left(\frac{n_0}{V}\right)\cdot \left(\frac{n_1}{V}\right)
 \frac{V^{n_0 + n_1}}{n_0! n_1 !}  \exp\left(
\mu_0 n_0 + \mu_1 n_1 - \frac{a}{V}n_0 n_1\right).
\end{eqnarray*}
That is, the limiting state in this case is the symmetric mixture
(convex combination with equal coefficients) of two pure states
(phases, see \cite[Chapter 7]{GeM}), say $P^{\pm}$. The particle
densities $\varrho_i^{\pm}$ in these phases are
\begin{equation*}
\varrho_0^{\pm} = u(a, \mu \pm \bar{y}(a,\mu)), \qquad
\varrho_1^{\pm} = u(a, \mu \mp \bar{y}(a,\mu)).
\end{equation*}
For these phases, like in the case of $\mu_0\neq \mu_1$ we get, see
(\ref{K6}) and (\ref{K7}), that
\begin{eqnarray*}
{\rm for} \ P^{+} \ \ & &\tilde{\mu}_0^\Lambda \to \mu-  u(a, \mu -
\bar{y}(a,\mu))= \ln \tilde{z}^{+}, \quad \ \tilde{\mu}_1^\Lambda
\to \mu- u(a, \mu + \bar{y}(a,\mu)) = \ln \tilde{z}^{-},
\\[.2cm]
{\rm for} \ P^{-} \ \ & &\tilde{\mu}_0^\Lambda \to \mu- u(a, \mu +
\bar{y}(a,\mu))= \ln \tilde{z}^{-}, \quad \ \tilde{\mu}_1^\Lambda
\to \mu- u(a, \mu -\bar{y}(a,\mu))= \ln \tilde{z}^{+}.
\end{eqnarray*}
By (\ref{K5}) and (\ref{K2}) this yields that $P^{+} =
P_{\tilde{z}^{+}, \tilde{z}^{-}}$ and $P^{-} = P_{\tilde{z}^{-},
\tilde{z}^{+}}$, which proves claim (ii). To prove claim (iii) we
use the first line in (\ref{4}) and then (\ref{KL11}) (resp.
(\ref{KL12})) with $\phi_V \equiv 1$ for $\mu_0 > \mu_1$ (resp.
$\mu_0 = \mu_1$). In both cases,  by (\ref{E1a}) this leads to
(\ref{9a}). \hfill$\square$

\section*{Acknowledgment}
The present research was supported by the European Commission under
the project STREVCOMS PIRSES-2013-612669. The first named author was
also supported by National Science Centre, Poland, grant
2017/25/B/ST1/00051.

\end{document}